\documentclass[final]{article}
\pdfoutput=1
\usepackage{fullpage}
\usepackage{lmodern}
\usepackage[utf8]{inputenc}
\usepackage{packages}

\usepackage{authblk}
\usepackage[backend=bibtex,style=numeric-comp,doi=false,url=false,firstinits=true]{biblatex}
\theoremstyle{plain}
\newtheorem{lemma}{Lemma}[section]
\newtheorem{theorem}[lemma]{Theorem}
\newtheorem{corollary}[lemma]{Corollary}
\newtheorem{proposition}[lemma]{Proposition}
\theoremstyle{definition}
\newtheorem{definition}[lemma]{Definition}
\newtheorem{assumption}[lemma]{Assumption}
\theoremstyle{remark}
\newtheorem{remark}[lemma]{Remark}
\newtheorem{example}[lemma]{Example}

\title{Stability of local quantum dissipative systems}

\author[1,2]{Toby S.\ Cubitt%
  \thanks{tsc25@cam.ac.uk}}
\author[1]{Angelo Lucia%
  \thanks{anlucia@ucm.es}}
\author[3]{Spyridon Michalakis%
  \thanks{spiros@caltech.edu}}
\author[1]{David Perez-Garcia%
  \thanks{dperezga@ucm.es}}
\affil[1]{Departamento de An\'alisis Matem\'atico, Universidad Complutense
  de Madrid, \authorcr 28040 Madrid, Spain}
\affil[2]{DAMTP, University of Cambridge, Centre for Mathematical Sciences, \authorcr
	Cambridge CB3 0WA, United Kingdom}
\affil[3]{Institute for Quantum Information and Matter, Caltech, Pasadena, CA 91125, USA}
\date{}

\begin{document}
\maketitle
\begin{abstract}
Open quantum systems weakly coupled to the environment are modeled by completely positive, trace preserving semigroups of linear maps. The generators of such evolutions are called Lindbladians.
In the setting of quantum many-body systems on a lattice it is natural
to consider Lindbladians that decompose into a sum of local interactions
with decreasing strength with respect to the size of their support.
For both practical and theoretical reasons, it is crucial to estimate the impact that perturbations in the generating Lindbladian, arising as noise or errors, can have on the evolution. These local perturbations are
potentially unbounded, but constrained to respect the underlying
lattice structure.
We show that even for polynomially decaying errors in the Lindbladian, local observables and correlation functions are stable if the unperturbed Lindbladian has a unique fixed point and a mixing time which scales logarithmically with the system size.
The proof relies on Lieb-Robinson bounds, which
describe a finite group velocity for propagation of information in local systems.
As a main example, we prove that classical Glauber dynamics is stable under local perturbations, including perturbations in the transition rates which may not preserve \emph{detailed balance}.

\end{abstract}

\tableofcontents

\section{Background and previous work}

The physical properties of a closed many-body quantum system are encoded in its Hamiltonian. Theoretical models of such systems typically assume some form of local structure, whereby the Hamiltonian decomposes into a sum over interactions between subsets of nearby particles. Similarly, the behavior of an open many-body quantum system is encoded in its Liouvillian. Again, this is typically assumed to have a local structure, decomposing into a sum over local Liouvillians acting on subsets of nearby particles.

Crucial to justifying such theoretical models is the question of whether their physical properties are stable under small perturbations to the local interactions. If the physical properties of a many-body Hamiltonian or Liouvillian depend sensitively on the precise mathematical form of those local terms, then it is difficult to conclude anything about physical systems, whose interactions will always deviate somewhat from theory.

Quantum information theory has motivated another perspective on many-body Hamiltonians. Rather than studying models of naturally occurring systems, it studies how many-body systems can be engineered to produce desirable behavior, such as long-term storage of information in quantum memories \cite{ToricCode4D,Haah11,RMP.82.1041,Maurer08062012}, processing of quantum information for quantum computing \cite{Kitaev20032,RMP.80.1083,RMP.82.1209,briegel2009measurement,Farhi20042001}, or simulation of other quantum systems which are computationally intractable by classical means \cite{bloch2012quantum,blatt2012quantum,aspuru2012photonic,houck2012chip,jordan2012quantum}. Again, stability of these systems under local perturbations is crucial, otherwise even tiny imperfections may destroy the desired properties. Stability in this context has been studied for self-correcting topological quantum memories, where one in addition requires robustness against local sources of \emph{dissipative} noise, and the relevant quantity is the minimum time needed to introduce logical errors in the system. It has been known since \cite{ToricCode4D,Toric4D} that a self-correcting quantum memory with local interactions is possible in four spatial dimensions. With the breakthrough of the Haah code \cite{Haah11}, it seems it may be possible to engineer such self-correcting quantum memories in three dimensions.

Recently, and partially motivated by the dissipative nature of noise, this ``engineering'' approach has been extended to open quantum systems and many-body Liouvillians. First theoretical \cite{Kraus08,verstraete09}, and then experimental \cite{PRL.107.080503,2010NatPh6.943B} work has shown that creating many-body quantum states as fixed points of engineered, dissipative Markovian evolutions can be more robust against undesirable errors and maintain coherence of quantum information for longer times. Intuitively, there is an inherent robustness in such models: the target state is independent of the initial state. If the dissipation is engineered perfectly, the system will always be driven back towards the desired state. This idea can be used to engineer dissipative systems both for storing quantum information \cite{Kraus08,verstraete09} and for carrying out computation via dissipative dynamics \cite{verstraete09}. However, it does not guarantee stability against errors \emph{in the engineered Liouvillian itself}. Once again, stability against local perturbations -- this time for many-body Liouvillians rather than Hamiltonians -- is of crucial importance.

In the case of closed systems governed by Hamiltonians, recent breakthroughs have given rigorous mathematical justification to our intuition that the physical properties of many-body Hamiltonians are stable. Starting with \cite{Klich:2010,Bravyi-Hastings-Spiros}, it culminated in the work of \cite{Spiros11} which showed that, under a set of mathematically well-defined and physically reasonable conditions, gapped many-body Hamiltonians are stable under perturbations to the local interactions.\footnote{Note that, in stark contrast to traditional perturbation theory, the perturbations considered here simultaneously change \emph{all} the local interactions by a small amount. The strength of the total perturbation therefore scales with system size and standard perturbation theory does not apply. It is the structure of local ground states of the Hamiltonian that ensures stability.} More precisely, in the presence of \emph{frustration-freeness}, \emph{local topological quantum order}, and \emph{local gap}, the spectral gap of a Hamiltonian with local (or quasi-local) interactions is stable against small (quasi-) local perturbations (see \cite{Spiros11} for a formal definition of these conditions). The bound on the amount of imperfection tolerated by the system depends on the decay of the local gaps, the decay of the local topological order, and the strength (and decay rate) of the interactions. Furthermore, except for frustration freeness which is a technical condition required in the proof, these conditions are in a sense tight. There exist simple counterexamples to stability if any one of the conditions is lifted.

\section{Stability of open quantum systems}

In this work, we study stability of many-body Liouvillians. We consider dynamics generated by rapidly decaying interactions, where the notion of rapid decay is made precise in section~\ref{sec:notation}. Moreover, in order to have a well-defined notion of scaling with system size, we restrict to Liouvillians whose local terms depend only on the subsystem on which they act,
and thus are not redefined as we consider larger systems. We call such families of Liouvillians \emph{uniform}.

Our main result shows that, under the above assumptions on the structure of the Liouvillian, logarithmic mixing time implies the desired stability in the dissipative setting.

However, although the result is analogous to \cite{Spiros11}, the proof and even the definition of stability in the case of Liouvillians necessarily differ substantially from the Hamiltonian case. For Hamiltonians, the relevant issue is stability of the spectral gap. Via the quasi-adiabatic technique \cite{quasi-adiabatic-1,quasi-adiabatic-2}, this in turn implies a smooth transition between the initial and perturbed ground states, showing that both are within the same phase. Note that the existence of a smooth transition (no closing of the spectral gap in the thermodynamic limit) does not imply that both ground states are close in norm, as the simple example $H = \sum_{i=1}^N {\proj{0}}_i$ vs.\ $H(\epsilon) = \sum_{i=1}^N (\ket{0}+\epsilon\ket{1})(\bra{0}+\epsilon\bra{1})_i/(1+\epsilon^2)$ shows.\footnote{Note that each Hamiltonian is the sum of non-interacting projections for any $\epsilon \in \R$. In particular, for each $\epsilon$, there is a unitary $U(\epsilon)$ acting on a single site, such that $H(\epsilon) = U(\epsilon)^{\otimes N} H U^{\dagger}(\epsilon)^{\otimes N}$.} It does however imply a well-behaved perturbation in the expectation value of local observables -- such as order parameters -- and correlation functions, which in most experimental situations are the only measurable quantities.

For Liouvillians, we are interested in a definition of stability more related to the evolution itself, which accounts at the same time for both the speed of convergence and the properties of the fixed point. Here, we consider the strongest definition of stability: we want our systems (initial and perturbed) to evolve similarly for all times and all possible initial states. Thus, not only should the speed of convergence to the fixed points be similar, the fixed points themselves should be close and so should the approach to the fixed points.

This definition is significantly stronger than stability of the spectral gap alone,\footnote{Due to the recent work in \cite{spectralboundsquantum}, it is not clear whether the spectral gap in Liouvillians is the relevant quantity for convergence questions.} and is more directly relevant to the applications discussed above. As in the Hamiltonian case, the analogous simple example shows that one cannot expect to attain such stability if we consider global measurements on the system. Therefore, in analogy with the Hamiltonian case, restrict our attention to local observables and few-body correlation functions. Since there are technical subtleties involved in extending this stronger definition of stability to dynamics with multiple fixed points, we defer consideration of multiple fixed points to a future paper, and restrict our attention here to dissipative dynamics with unique fixed points. It is important to note, however, that we do not make \emph{any} assumption on the form of the unique fixed point. In particular, we do not assume that it is full-rank (primitivity); our results apply equally well to Liouvillians with pure fixed points. (Pure-state fixed points are particularly relevant to quantum information applications, such as dissipative state engineering and computation.)

A key technical ingredient in the stability proof for Hamiltonians is the quasi-adiabatic evolution technique \cite{quasi-adiabatic-1,quasi-adiabatic-2}, which directly uses the fact that Hamiltonian evolution is reversible. This is of course no longer true for Liouvillians, so we must use a different proof approach. We make use of the fact that evolution under a Liouvillian converges to a steady-state, together with dissipative generalizations \cite{Nachtergaele12} of the Lieb-Robinson bounds that are the other crucial ingredient in \cite{Spiros11}.

Among systems which satisfy our assumptions, one finds classical Glauber dynamics \cite{Martinelli97}. This immediately shows that Glauber dynamics is stable against errors. To the best of our knowledge, this is new even to the classical literature (related results, but with different assumptions, were given in \cite{MR814713}). Given the importance of Glauber dynamics to sampling from the thermal distributions of classical spin systems \cite{Liggett85,Martinelli97}, we expect our results to have applications also to classical statistical mechanics.

The paper is structured as follows: After setting up notation and basic definitions in the next section, we state our main stability result in section~\ref{sec:main-result} and discuss the assumptions it requires. In section~\ref{sec:toolbox} we prove various technical results used in the main proof, which is given in section~\ref{sec:proof}. We apply these results in section~\ref{sec:glauber-dynamics} to the important example of classical Glauber dynamics, before concluding with a discussion of the results and related open questions in section~\ref{sec:conclusions}.


\section{Setup and notation}\label{sec:notation}
We will consider a cubic lattice\footnote{We restrict to cubic lattices for the sake of exposition. The results can be extended to more general settings, replacing the lattice $\Z^D$ with a graph with polynomial growth.} $\Gamma = \Z^D$. The ball centered at $x \in \Lambda$ of radius $r$ will be denoted by $b_x(r)$. At each site $x$ of the lattice we will associate one elementary quantum system with a finite dimensional Hilbert space $\mathcal H_x$. We will use the Dirac notation for
vectors: $\ket{\phi}$ will denote a vector in $\mathcal H_x$, $\bra{\phi}$ its adjoint, and  $\{ \ket{n} \}_{n=0}^{\dim \mathcal H_x}$ the canonical basis for $\mathcal H_x$. Scalar product in $\mathcal H_x$ will be denoted by $\braket{\phi}{\psi}$, and rank-one linear maps by $\ketbra{}{\phi}{\psi}$.
For each finite subset $\Lambda \subseteq \Gamma$, the associated Hilbert space is given by
\begin{equation}
\label{eq:hilbert-space}
 \mathcal H_\Lambda = \bigotimes_{x \in \Lambda} \mathcal H_x ,
\end{equation}
and the algebra of observables supported on $\Lambda$ is defined by
\[ \mathcal A_\Lambda = \bigotimes_{x \in \Lambda} \mathcal B(\mathcal H_x) .\]
If $\Lambda_1\subset \Lambda_2$, there is a natural inclusion of $\mathcal A_{\Lambda_1}$ in $\mathcal A_{\Lambda_2}$ by identifying it with $\mathcal A_{\Lambda_1}\otimes \identity$. The support of an observable $O \in \mathcal A_\Lambda$ is the minimal set $\Lambda^\prime$ such that $O = O^\prime\otimes \identity$, for some $O^\prime \in \mathcal A_{\Lambda^\prime}$, and will be denoted by $\supp O$. We will denote by $\norm{\cdot}_p$ the Schatten $p$-norm over $\mathcal A_\Lambda$. Where there is no risk of ambiguity, $\norm{\cdot}$ will denote the usual operator norm (i.e.\ the Schatten $\infty$-norm).

A linear map $\mathcal T: \mathcal A_\Lambda \to \mathcal A_\Lambda$ will be called a \emph{superoperator} to distinguish it from operators acting on states. The support of a superoperator $\mathcal T$ is the minimal set $\Lambda^\prime \subseteq \Lambda$ such that $\mathcal T = \mathcal T^\prime \otimes \identity$, where $\mathcal T^\prime \in \mathcal B(\mathcal A_{\Lambda^\prime})$. A superoperator is said to be Hermiticity preserving if it maps Hermitian operators to Hermitian operators. It is said to be positive if it maps positive operators (i.e. operators of the form $\du O O$) to positive operators. $\mathcal T$ is called \emph{completely positive} if $\mathcal T \otimes \identity : \mathcal A_\Lambda \otimes M_n \to \mathcal A_\Lambda \otimes M_n$ is positive for all $n \ge 1$. Finally, we say that $\mathcal T$ is trace preserving if $\trace \mathcal T(\rho) =\trace \rho$ for all $\rho \in A_\Lambda$. For a general review on superoperators, see \cite{Wolf11}.

The dynamics of the system is generated by a superoperator $\mathcal L$, which plays a similar role to the Hamiltonian in the non-dissipative case. The evolution will be given by the one parameter semigroup $T_t = e^{t \mathcal L}$. The natural assumptions to make about $T_t$ are that it is a continuous semigroup of completely positive and trace preserving maps (CPTP, sometimes also called \emph{quantum channels}). Such maps are always contractive, meaning that $\norm{T_t}_{1\to 1, cb} \le 1$, where the completely-bounded norm is defined as:
\begin{equation}
\label{eq:1-to-1}
 \norm{T}_{1\to 1, cb} = \sup_n \norm{T\otimes \identity_n}_{1\to 1} = \sup_n \sup_{ \substack{X \in \mathcal A_\Lambda \otimes M_n \\ X \neq 0 }} \frac{ \norm{T\otimes \identity_n(X)}_1}{\norm{X}_1} .
\end{equation}
We will also be interested in the $\norm{\cdot}_{\infty \to \infty, cb}$ completely-bounded norm of superoperators, which is defined as follows:
\[ \norm{T}_{\infty\to \infty, cb} = \sup_n \norm{T\otimes \identity_n}_{\infty\to \infty} = \sup_n \sup_{\substack{X \in \mathcal A_\Lambda \otimes M_n \\ X \neq 0 }} \frac{ \norm{T\otimes \identity_n(X)}_\infty}{\norm{X}_\infty} .\]
The relationship between $\norm{\cdot}_{1\to 1,cb}$ and $\norm{\cdot}_{\infty \to \infty, cb}$ is the following:
\[ \norm{T}_{1 \to 1, cb}  = \norm{\du T}_{\infty \to \infty, cb},\]
where $\du T$ is the dual of $T$, satisfying $\trace A\, T(B) = \trace \du T(A)\, B$.
We will denote $\norm{\cdot}_{\infty\to\infty, cb}$ simply by $\norm{\cdot}_{cb}$ when there is no risk of confusing different completely-bounded norms.

\begin{remark}
\label{obs:stabilized-cbnorm}
As shown in \cite{stabilizedCBnorm}, the supremum in equation~\eqref{eq:1-to-1} is reached when $n$ is equal to the dimension of the space on which $T$ is acting: if $T : \mathcal M_n \to \mathcal M_n$, then $\norm{T\otimes \identity_n}_{1\to1} = \norm{T}_{1\to 1,cb}$.
\end{remark}

The generator $\mathcal L$ of the semigroup $T_t = e^{t \mathcal L}$, is called a \emph{Liouvillian}. All such generators can be written in the following general form, often called the \emph{Lindblad form}~\cite{Kossakowski,Lindblad} (see \cite{Wolf11}):
\begin{proposition}
$\mathcal L$ generates a continuous semigroup of CPTP maps if and only if it can be written in the form:
\begin{equation}\label{lindblad_form}
\mathcal L(\rho) = i [\rho, H] + \sum_j L_j \rho \du L_j - \frac 1 2 \sum_j \{\du L_j L_j, \rho \},
\end{equation}
where $H$ is a Hermitian matrix, $\{ L_j \}_j$ a set of matrices called the Lindblad operators, $ [\cdot,\cdot]$ denotes the commutator and $\{ \cdot, \cdot \}$ the anticommutator.
\end{proposition}
We will use the term \emph{Lindbladian} and Liouvillian interchangeably.
Since we consider Lindbladians $\mathcal L$ corresponding to local dissipative dynamics, we assume that $\mathcal L$ is a \emph{local Lindbladian} of the form:
\begin{equation}
\mathcal L = \sum_{u\in \Lambda} \sum_{r\ge 0} \mathcal L_{u,r} , \quad \supp \mathcal L_{u,r} = b_u(r),
\end{equation}
where each term in the sum above can be written in the form given by equation~\eqref{lindblad_form}.

Such a decomposition is obviously always trivially possible. We are interested in the cases in which the norms of $\mathcal L_{u,r}$ decay with $r$. Concretely, let us define the \emph{strength of interaction} for a Lindbladian as the pair $(J, f)$ given by:
\begin{equation}
J = \sup_{u,r} \norm{\mathcal L_{u,r}}_{1\to 1,cb}, \quad f(r) = \sup_u \frac{\norm{\mathcal L_{u,r}}_{1 \to 1,cb}}{J} .
\end{equation}

The behavior of $f(r)$ as $r$ goes to infinity corresponds to various interaction regimes, listed in order of decreasing decay rate:
\begin{itemize}
\item finite range interaction: $f(r)$ is compactly supported;
\item exponentially decaying: $ f(r) \le e^{-\mu r}$, for some $\mu > 0$;
\item quasi-local interaction: $f(r)$ decays faster than any polynomial;
\item power-law decay: $f(r) \le (1 +r)^{-\alpha}$, for some positive $\alpha >0$.
\end{itemize}

As we will see later, our result will apply whenever $\mathcal L$ has finite range, exponentially decaying, or quasi-local interactions. It will also hold in the power-law decay regime, but we will require a lower bound on the decay exponent $\alpha$, depending on the dimension of the underlying lattice. Not to overload the exposition, we will assume that $\mathcal L$ has finite range or exponentially decaying interactions, unless otherwise specified. The modifications needed to work with quasi-local interactions and power-law decay are presented in section~\ref{sec:power-law-decay}. Also, we will say that functions we construct along the way are \emph{fast-decaying}, if their decay rate is within the same decay class of $f(r)$ we are considering (or faster).

As shown in \cite{inverseeig}, from the spectral decomposition of $\mathcal L$ (and $T_t$) one can define two new CPTP maps which represent the infinite-time limit of the semigroup $T_t$. We will denote by $T_\infty$ the projector onto the subspace of stationary states (fixed points), and by $T_\phi$ the projector onto the subspace of periodic states. They correspond, respectively, to the kernel of $\mathcal L$ and to the eigenspace of purely imaginary eigenvalues of $\mathcal L$, which we denote $\mathcal F_{\mathcal L}$ and $\mathcal X_{\mathcal L}$, respectively. Both subspaces are invariant under $T_t$: in particular, $T_t$ acts as the identity over $\mathcal F_{\mathcal L}$, while it is a unitary operator over $\mathcal X_{\mathcal L}$. Note, also, that both subspaces are spanned by positive operators (i.e.\ density matrices)~\cite[Prop. 6.8, Prop. 6.12]{Wolf11}.
We will denote by $T_{\phi,t}$ the composition $T_t \circ T_\phi$.

Since we plan to exploit the local structure of $\mathcal L$, we will often make use of the restriction of $\mathcal L$ to a subset of the lattice. Given $A \subset \Lambda$, we define the \emph{truncated}, or \emph{localized}, generator:
\begin{equation}
  \mathcal L_A = \sum_{b_u(r) \subseteq A} \mathcal L_{u,r} .
\end{equation}

\subsection{Uniform families}
We are interested in how properties of dissipative dynamics scale with the size of the system. Hence, we are concerned with sequences of Lindbladians defined on lattices of increasing size, where all the Lindbladians in the sequence are from the same ``family''. To make this precise, we need to pin down how Lindbladians from the same family, but on different size lattices, are related to one-another. Our results will apply to very general sequences of Lindbladians, which we call \emph{uniform families}. Before giving the precise definition, it is helpful to consider some special cases.

For local Hamiltonians on a lattice, one often considers translationally-invariant systems with various types of boundary conditions (e.g.\ open or periodic boundaries). There is then a natural definition of what it means to consider the same translationally-invariant Hamiltonian on different lattice sizes. Translationally-invariant Lindbladians are an important special case of a uniform family. In this special case, all the local terms in the Lindbladian that act in the ``bulk'' of the lattice are the same. Another way of thinking about this is to formally consider the translationally-invariant Lindbladian $\mathcal M$ defined on the infinite lattice $\Gamma=\Z^d$, and then consider each member of the family to be a restriction of this infinite Lindbladian to a finite sub-lattice $\Lambda\subset\Gamma$ of some particular size:
\[ \mathcal L = \mathcal M_\Lambda .\]

This gives us translationally-invariant Lindbladians with \emph{open boundary condition}. But of course, this is only one particular choice of boundary terms (in this case, no boundary terms at all). We are also interested in more general boundary conditions, such as periodic boundaries. So, in addition to the ``bulk'' interactions coming from $\mathcal M$, we allow additional terms that play the role of \emph{boundary conditions}:
\[ \mathcal L = \mathcal M_\Lambda + \mathcal L^{\partial \Lambda} .\]

We allow greater freedom in the boundary terms $\mathcal L^{\partial \Lambda}$. For one thing, they are allowed to depend on the size of the lattice $\Lambda$. But more importantly, we allow strong interactions that \emph{cross the boundary} of $\Lambda$, coupling sites that would otherwise be far apart. For example, the case of \emph{periodic boundary conditions} corresponds to adding interaction terms that connect opposite boundaries of $\Lambda$, as if on a torus (see Fig.~\ref{fig:boundary_term}).

Now that we have given an intuition of what a uniform family is, it is time to present the formal definition. This includes all the special cases discussed so far, but also captures much more general families of Lindbladians that are not necessarily translationally-invariant, and many other types of boundary conditions (e.g.\ cylindrical boundaries, or boundary terms that give the sphere topology, or terms that force fixed states on the boundary\footnote{Or even M\"obius strips, Klein bottles, and other exotic topologies.}).

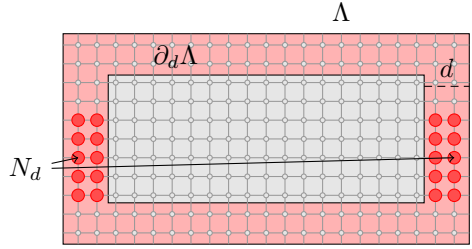
\begin{figure}
\begin{center}
  \begin{tikzpicture}
    \node (A0) at (-2.7,-1.4) {};
    \node (A1) at ($ -1*(A0) $) {};
    \node (B0) at (-2.1, -0.85) {};
    \node (B1) at ($ -1*(B0) $) {};
    \tikzstyle{boundarynode}=[draw=red, fill=red!70, circle, scale=0.5];
    \tikzstyle{bulknode}=[draw=black!40, fill=black!10, circle, scale=0.2];
    \draw[fill=red!30] (A0) rectangle (A1);
    \draw[fill=black!10] (B0) rectangle (B1);
    \path[step=.25,draw=black!40] (A0) grid (A1);
    \path[densely dashed,draw] (2.1,0.7) -- node[above] {$d$} ++(0.6,0);

    \foreach \x in {-2.5, -2.25, ..., 2.5}
    \foreach \y in {-1.25, -1, ..., 1.25}
      \node[bulknode] at (\x,\y) {};

    \foreach \y in {-2.5, -2.25}
    \foreach \x in {-0.75, -0.5, ..., 0.25} {
      \node[boundarynode] at (\y,\x) {};
      \node[boundarynode] at (-\y,\x) {};
    }

    \node at (1,1.7) {$\Lambda$};
    \node at (-1.2,1.1) {$\partial_d \Lambda$};

    \node (ND) at ($ (A0) + (-0.5,1) $) {$N_d$};
    \draw[->] (ND) -- (-2.5, -0.25) ;
    \draw[->] (ND) -- (2.5, -0.25) ;

  \end{tikzpicture}
\caption{{\small{Partition of the lattice $\Lambda$ into the bulk and the boundary of thickness $d$, $\partial_d \Lambda$ (see Def.~\ref{defn:boundary-condition}). The dark red regions on the boundary correspond to the interaction term $N_d$ coupling distant regions in $\Lambda$.}}}
\label{fig:boundary_term}
\end{center}
\end{figure}

\begin{definition}\label{defn:boundary-condition}
  Given $\Lambda \subset \Gamma$, a \emph{boundary condition} with strength $(J, f)$ for $\Lambda$ is a Lindbladian $\mathcal L^{\partial \Lambda} = \sum_{d\ge 1} N_d$, where
  \begin{gather*}
    \norm{N_d}_{1\to 1, c.b.} \le J \abs{\partial_d \Lambda} f(d)
    \intertext{with}
    \partial_d \Lambda :=  \{ x \in \Lambda \,|\, \dist(x,\Lambda^c) \le d \},\\
    \supp N_d \subset \partial_d \Lambda.
  \end{gather*}
\end{definition}

\begin{definition}
  \label{def:uniform-family}
  A \emph{uniform family} of Lindbladians $\mathcal L$ with strength $(J,f)$ is given by the following:
  \begin{enumerate}[(i)]
  \item \textit{infinite Lindbladian}: a Lindbladian $\mathcal M$ on all of $\Z^D$ with strength $(J,f)$;
  \item \textit{boundary conditions}: a set of {\it boundary conditions} $\mathcal L^{\partial \Lambda}$, with strength $(J, f)$ and $\Lambda = b_u(L)$, for each $u\in \Z^D$ and $L \ge 0$.
  \end{enumerate}

  We say that the family is \emph{translationally invariant} if $\mathcal M$ is translationally invariant and $\mathcal L^{\partial b_u(L)}$ is independent of $u$.
\end{definition}

Given a uniform family $\mathcal L$, we fix the following notation for evolutions defined on a subset $\Lambda$:
\begin{align}
\mathcal L^\Lambda = \mathcal M_\Lambda &\quad \text{``open boundary'' evolution} ; \\
\mathcal L^{\overline \Lambda} = \mathcal M_\Lambda + \mathcal L^{\partial \Lambda} &\quad \text{``closed boundary'' evolution},
\end{align}
with the respective evolutions $T_t^\Lambda = \exp(t \mathcal L^\Lambda)$ and  $T_t^{\overline \Lambda} = \exp(t \mathcal L^{\overline \Lambda})$.

\begin{remark}
  \label{obs:disjoint-region}
  In the rest of the paper, we will make use of the following notation:
  \[A(s) = \{ x \in \Lambda \,|\, \dist(x,A) \le s \}.\]
  Since we are interested in observables whose support is not connected, we want to consider more general regions than balls: in particular, we are interested in disjoint unions of convex regions (for example, to calculate two-point correlation functions).
  Consider what happens to such a region $A = A_0 \sqcup A_1$ when we grow it by taking $A(s)$.
  When $s$ becomes sufficiently large, $A_0(s)$ will merge with $A_1(s)$.
  At this point, $A(s)$ will not be a disjoint union of balls anymore.
  To avoid such complications, for $s$ large enough that disjoint balls merge, we will replace $A(s)$ by the smallest ball containing it.
  This will not hurt us, as $\abs{A(s)}$ will still grow asymptotically at the same rate, which will be sufficient for our purposes.
\begin{figure}[h]
\begin{center}
  \begin{tikzpicture}
    \path[step=.25cm,draw=black!40] (-2.5,-2.5) grid (2.5,2.5);
    \draw[thick,pattern=north west lines,pattern color=cyan] (0,0) circle(2cm);
    \foreach \x/\label in {-1/A_0, 1/A_1} {
      \node (\label) at (-\x,0) {};
      \draw[thick,fill=cyan!30] (\label) circle(1cm);
      \draw[thick,fill=magenta!50] (\label) circle(0.4cm);
      \node at (\label) {$\label$};
    };
    \path[dashed,draw] ($(A_0) - (0.4,0)$) -- node[above] {$s$} ($ (A_0) - (1,0) $);
    \path[dashed,draw] ($(A_1) + (0.4,0)$) -- node[above] {$s$} ($ (A_1) + (1,0) $);
    \node at (0,1.3) {$A(s)$};
  \end{tikzpicture}
  \caption{The convention on how to grow a region $A=A_0 \sqcup A_1$.}
\end{center}
\end{figure}
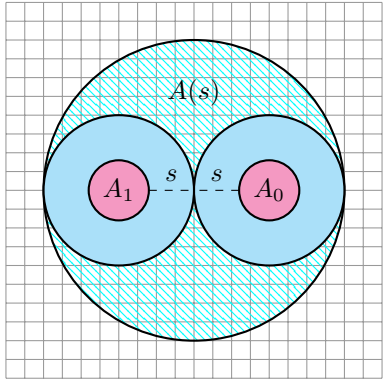
\end{remark}

\begin{definition}\label{defn:fixed-point}
  We say that $\mathcal L$ has a unique fixed point if, for all $\Lambda = b_u(L)$, $\mathcal X_{\mathcal L^{\overline \Lambda}} = \mathcal F_{\mathcal L^{\overline \Lambda}} = \{ \rho_\infty^{\overline \Lambda} \}$. In other words, $T^{\overline \Lambda}_\phi(\rho) = T^{\overline \Lambda}_\infty(\rho)=\rho^{\overline \Lambda}_\infty$, for all density matrices $\rho$.
\end{definition}

Note that if for all pure states $\rho$, we have $T^{\overline \Lambda}_t(\rho) > 0$ (positive definite), for $t > 0$, then the evolution has a unique fixed point $\rho_\infty > 0$ (see \cite[Thm. 6.7]{Wolf11}).

We will drop the superscript from $T^{\overline \Lambda}_t$, and simply write $T_t$, when we consider some fixed $\Lambda \subset \Gamma$. In that case, we will refer to the number of lattice sites in $\Lambda$ as the \emph{system size}.


\section{Main result}
\label{sec:main-result}
\subsection{Assumptions for stability}

In Hamiltonian systems, the spectral gap (the difference between the two lowest energy levels) plays a crucial role in a number of settings, from defining quantum phases and phase transitions \cite{quantphase} to understanding the entanglement and correlations present in the system \cite{Hastings06,Hastings07,Hastings10} and analyzing its stability to perturbations~\cite{Bravyi-Hastings-Spiros,Spiros11}. On the other hand, it is known that for Lindbladians, the spectral gap (in this setting, the least negative real part of the non-zero eigenvalues) alone is not sufficient to fully characterize the convergence properties of the dissipative evolution~\cite{Kastoryano12,spectralboundsquantum}. Therefore, we will instead impose a more general requirement on the convergence of the dynamics. (The dependence of this requirement on spectral properties of $\mathcal L$, i.e.\ properties depending on the eigenvalues -- like the gap -- and eigenvectors -- like the condition number, is an active area of research.)

\begin{definition}[rapid mixing]
\label{def:global-rapid-mixing}
Given a one-parameter semigroup of CPTP maps $T_t$, define the \emph{contraction} of $T_t$ as the following quantity:
\begin{equation}
  \eta(T_t) = \frac{1}{2} \sup_{\substack{\rho \ge 0 \\ \trace \rho\, =1 }} \norm{T_t(\rho)- T_{\phi,t} (\rho)}_1.
\end{equation}

Given a family of such semigroups $\{T_t^\alpha \}_\alpha$, each of which is acting on $\mathcal B( H_{\alpha})$
for some Hilbert space $H_{\alpha}$ of finite dimension $d_\alpha$,
we say that it satisfies \emph{rapid mixing} if there exist constants $c, \gamma, \delta >0$,
such that for each $\alpha$:
\begin{equation}
\label{eq:global-rapid-mixing}
  \eta(T_t^\alpha) \le c \log^\delta(d_\alpha) \ e^{-t \gamma} .
\end{equation}
We will write RM$(\gamma, \delta)$ for short.
\end{definition}

If each $H_\alpha$ has a tensor product structure of the type defined in equation~\eqref{eq:hilbert-space},
then the rapid mixing assumption can be restated as a logarithmic scaling with system size of the mixing time.
Since the dimension of $\mathcal H_\Lambda$ is $(\dim \mathcal H_x)^{\abs{\Lambda}}$, for uniform families condition~\eqref{eq:global-rapid-mixing} is equivalent to:
\begin{equation}
\label{eq:non-ti-grm}
 \eta(T_t^{\overline \Lambda}) \le c \abs{\Lambda}^\delta e^{-t \gamma} \quad \forall \Lambda.
\end{equation}

Let us recall a result from \cite{Kastoryano12}.
\begin{theorem}[Contraction for commuting Lindbladians]\label{thm:commuting-liouvillians}
Let $\{ \mathcal L_j \}_{j=0}^n$ be a set of commuting Lindbladians. Define $\mathcal L = \sum_j \mathcal L_j$ and the corresponding evolutions $T^j_t = e^{t \mathcal L_j}$ and $T_t = e^{t \mathcal L}$. Then:
\begin{equation}
\eta(T_t) \le \sum_j \eta(T_t^j) .
\end{equation}
\end{theorem}

In particular, consider the definition of $T_t^{\overline \Lambda}$ given in remark~\ref{obs:disjoint-region}
for $\Lambda \subset \Gamma$ being a disjoint union of balls.
Then the previous theorem implies that, if $\mathcal L$ is translationally-invariant and
it satisfies equation~\eqref{eq:global-rapid-mixing} for each of the connected components of
$\Lambda$, then it also satisfies the same equation (up to constants) for~$\Lambda$.

Finally, for translationally-invariant uniform families of Lindbladians,
it is sufficient to satisfy equation~\eqref{eq:global-rapid-mixing} for lattices centered at the origin: $\Lambda = b_0(L)$, $L \ge 1$.

\subsection{Stability}
With the required assumptions laid out, we can now state our main result.

\begin{theorem}
  \label{thm:stability}
  Let $\mathcal L$ be a uniform family of local Lindbladians with a unique fixed point, satisfying rapid mixing (equation~\eqref{eq:non-ti-grm}), and consider a perturbation of the form:
$$E^{\overline \Lambda} =  \sum_{u\in \Lambda} \sum_{r \ge 0} E_{u,r} + \sum_{d\ge 1} E_d,$$
where $E_{u,r}$ is supported on $b_u(r)$ and each $E_d$ is supported on $\partial_d \Lambda$ (see definition~\ref{defn:boundary-condition}) and
\[ \norm{E_{u,r}}_{1\to 1,cb} \le \epsilon\, e(r), \quad \norm{E_d}_{1\to 1, cb} \le \epsilon \abs{\partial_d \Lambda} e(d), \]
where $\epsilon > 0$ is a constant (the {\it strength} of the perturbation) and $e(r)$ is a fast-decaying function.
Consider the perturbed evolution $$S_t = \exp t (\mathcal L^{\overline \Lambda} + E^{\overline \Lambda} )$$ and suppose that the following assumptions hold:
  \begin{enumerate}[(i)]
  \item $\du E_{u,r}(\identity) = \du E_d(\identity) = 0$ (or, equivalently: $\trace E_{u,r}(O_A) = \trace E_d(O_A) = 0$, for all operators $O_A$).
  \item $S_t$ is a contraction for each $t\ge 0$.
  \end{enumerate}
For an observable $O_A$ supported on $A  \subset \Lambda$, we have for all $t \ge 0$:
  \begin{equation}\label{eq:stability}
    \norm{T_t^{*}(O_A) - S_t^{*}(O_A)} \le c(\abs{A})\, \norm{O_A} \left( \epsilon + \abs{\Lambda} \nu^{-1}_\eta(d_A) \right),
  \end{equation}
  where $d_A = \dist(A, \Lambda^c)$;
  $\eta$ is positive and independent of $\Lambda$; $\nu_\eta^{-1}(d) \le (1+d)^{-D-1}$;
  $c(|A|)$ is independent of $\Lambda$ and $t$, and is bounded by a polynomial in $\abs{A}$.
\end{theorem}

\begin{remark}
Note that, for a fixed $A$, if we let $\Lambda$ grow then $d_A$ will increase
with the linear size of $\Lambda$ and consequently $\abs{\Lambda}\nu_\eta^{-1}(d_A)$ will vanish in the limit.
\end{remark}

\begin{remark}
The assumptions (i)-(ii) on the perturbation $E$ are satisfied whenever $\mathcal M_{u,r} + E_{u,r}$ and $ N_d + E_d$ (as in definition~\ref{def:uniform-family})
are Lindbladians, but  the theorem also covers more general perturbations.
\end{remark}

\begin{remark}
  Since we are free to choose an $O_A$ with support on two non connected regions, we can apply theorem~\ref{thm:stability-single-fixed-point} to two-point correlation functions (or more generally $k$-point correlation functions, for fixed $k$) and still obtain that the error introduced by the perturbation depends linearly on the strength of the perturbation (and not on its global norm).
\end{remark}

A set of tools already applied in the setting of classical Markov chains \cite{Gross-1,Gross-2,Gross-3,Martinelli97}, and recently generalized to quantum dissipative systems \cite{Quantum-Log-Sobolev,temme:122201}, are the so-called {\it Logarithmic Sobolev inequalities} (in short, log-Sobolev inequalities). Introduced in a different setting to study hypercontractivity of semigroups \cite{King}, they provide the right asymptotic regime needed to satisfy the rapid mixing condition: in fact, the existence of a system size independent log-Sobolev constant implies a logarithmic scaling of the mixing time, which is exactly what is required in definition~\ref{def:global-rapid-mixing}. Without going into the technical details of log-Sobolev inequalities (which can be found in \cite{Quantum-Log-Sobolev,temme:122201}), we summarize this fact in the following corollary:

\begin{corollary}
\label{cor:stability-log-sobolev}
Let $\mathcal L$ belong to a uniform family of translationally-invariant Lindbladians with a unique fixed point for each system size. If $\mathcal L$ satisfies the log-Sobolev inequality with a system-size independent constant, then the dissipative dynamics are stable, in the sense of theorem~\ref{thm:stability}.
\end{corollary}

In particular, in \cite{2014arXiv1403.5224T} it was shown that product evolutions, i.e.\ Lindbladians that can be decomposed as a sum of independent terms $\mathcal L_k$ acting on a single subsystem,
satisfy a log-Sobolev inequality with a log-Sobolev constant lower bounded by the minimum of the spectral gaps of $\mathcal L_k$ (times a factor depending on the maximum dimension of the subsystems).
Moreover, the authors of \cite{2014arXiv1403.5224T} were able to show that
Davies maps associated to a graph state Hamiltonian \cite{MR2427572}
(which are not in a product form, but can be analyzed in a similar way)
and the ones associated to free-fermionic Hamiltonians
have a system-size independent log-Sobolev constant.

In all such cases, corollary~\ref{cor:stability-log-sobolev} implies that the evolution of local observables is stable.

\subsection{Local observables vs.\ global observables}
The bound in equation~\eqref{eq:stability} scales with the size of the support of the observable $O_A$. Although the dependence is polynomial, for observables with large support the result is not useful. Still, in most realistic experiments, we are interested in the behavior of observables with fixed support and low-degree correlation functions, making the above result widely applicable. Nonetheless, one might ask more generally for a system-size independent bound on:
\begin{equation}\label{bnd:fixed_points}
\sup_{\rho} \| T_\infty(\rho)-  S_\infty(\rho)\|_1,
\end{equation}
where $S_\infty$ is the fixed-point projector for the evolution of the perturbed Lindbladian. However, this is not possible; the limitation to local observables is in some sense strict. There is no hope of finding such a bound for global observables, as the following simple example shows.\footnote{Indeed, all global stability results for quantum Linbladians we are aware of have a dependency on the total Hilbert space dimension \cite{Oleg-Wolf}.}
\begin{example}
  Consider $N$ independent amplitude damping processes, with uniform rate $\gamma$ (which we can suppose w.l.o.g.\ equal to $1$). This Lindbladian can be written as
  \[ \mathcal L_N = \sum_{k=1}^{N} \identity_{1\dots k-1} \otimes \mathcal L_1 \otimes \identity_{k+1 \dots N} ,\]
  where
\[ \mathcal L_1(\rho)= \ketbra{}{0}{1} \rho \ketbra{}{1}{0} -\frac 12 \{ \rho, \diagstate{1} \} \]
 is an amplitude damping process on a single qubit, describing the decay of the state $\ket{1}$ into $\ket{0}$ at a constant rate $\gamma =1$. This Lindabladian has gap $1/2$ and $e^{t \mathcal L_N} = (e^{t \mathcal L_1})^{\otimes N}$ has mixing time of order $O(\log N)$ \cite[Sec. V. C.]{Kastoryano12}. The fixed point is the pure state $\diagstate{0\dots0}$.

Now consider $\mathcal L^\epsilon_1$, a rotation of $\mathcal L_1$, which fixes $\ket{\alpha_0} = \sqrt{1-\epsilon^2}\ket{0}$ $+ \epsilon \ket{1}$. We have $\norm{\mathcal L_1 - \mathcal L^\epsilon_1}_{1\to 1} = O(\epsilon)$, but the new fixed point $\diagstate{\alpha_0}^{\otimes N}$ is almost orthogonal to the original one, since the overlap between the two is
\[ \braket{0\dots0}{\alpha_0\dots \alpha_0} = \braket{0}{\alpha_0}^N = (1-\epsilon^2)^{N/2} \sim e^{-N \epsilon^2/2} \to 0 \text{ as } N \to \infty.\]
This shows that, in general, there is no good bound on~\eqref{bnd:fixed_points} (note that we have $\|\diagstate{0\dots 0}-\diagstate{a_0\dots \alpha_0}\|_1 \ge 1- |\braket{0\dots 0}{\alpha_0\dots \alpha_0}|^2$) and that the dependence on the support of the observable in equation
\eqref{eq:stability} cannot be improved:
to see this consider the observable $O_r = \ketbra{1\dots r}{0 \dots 0}{0 \dots 0}$ acting on $r \le N$ spins. $O_r$ has norm one, and
\[ O_\infty := \lim_{t\to \infty}T_t^*(O_r) = \identity, \quad
O_\infty^\epsilon := \lim_{t \to \infty}T_t^{\epsilon\, *}(O_r) = \braket{0}{\alpha_0}^{2r} \identity = (1 - \epsilon^2)^r \identity.\]
Consequently, we have:
\[ \norm{O_\infty - O_\infty^\epsilon} = 1 - (1-\epsilon^2)^r = r \epsilon^2 + o(\epsilon^2) .\]
This implies that any upper bound to $\norm{O_\infty - O_\infty^\epsilon}$ has to be at least
linear in $r$, which is the size of the support of $O_r$.
\end{example}

\subsection{Do we need all the assumptions?}
It is reasonable to ask if the assumptions of theorem~\ref{thm:stability} are all necessary. We have just shown that we must necessarily consider local observables if we are to have meaningful bounds, but what about the other conditions? We will now present three
examples, each consisting of a family of Lindbladiands with periodic boundary conditions, such that, in order:
\begin{itemize}
  \item The family is uniform and translationally invariant, satisfies rapid mixing,
    but does not have a unique fixed point;
  \item The family has a unique fixed point,
    but is not uniform and fails to satisfy rapid mixing;
  \item The family (presented in appendix~\ref{sec:non-stable-example})
    has a unique fixed point, satisfies rapid mixing,
    but is not uniform.
\end{itemize}
All these systems will be shown to be unstable.

\begin{example}
  \label{ex:counterexample-1}
  Consider a 1D chain composed of $N$ 4-level systems, with an independent Lindbladian acting on each site, having the following Lindblad operators
  \[ L_1 = \ketbra{}{0}{1}, \quad
  L_2 = \ketbra{}{0}{3}, \quad
  L_3 = \ketbra{}{2}{1}, \quad
  L_4 = \ketbra{}{2}{3}, \]
  and denote by
  \[ \mathcal L_0(\rho) = \sum_{i=1}^4 L_i \rho \du{L_i} - \frac 12 \{ \rho, \du{L_i} L_i \} .\]
  The global Lindbladian $\mathcal L_N$ is given by applying $\mathcal L_0$ independently on each site $k=1\dots N$:
  \[ \mathcal L_N = \sum_{k=1}^N \identity_{1,\dots, k-1} \otimes \mathcal L_0 \otimes \identity_{k+1,\dots, N} .\]
  Then we have that
  \[ \mathcal L_0( \ketbra{}{i}{j} ) =
  \begin{cases}
    0 & \text{if } i=j \in \{ 0, 2 \} \\
    \diagstate{0} + \diagstate{2} - 2 \ketbra{}{i}{j} & \text{if } i=j \in \{1,3\} \\
    - [\chi_{ \{1,3\} } (i) + \chi_{ \{1,3\} }(j)] \ketbra{}{i}{j} & \text{if } i\neq j .
  \end{cases}\]

  Diagonal states of the form $\diagstate{i}$ evolve according to the classical Markov process embedded in the Lindbladian, while off-diagonal elements $\ketbra{}{i}{j}$ evolve as
  \[ T_t(\ketbra{}{i}{j} )= \exp(-t [\chi_{ \{1,3\} } (i) + \chi_{ \{1,3\} }(j)] ) \ketbra{}{i}{j} ;\]
  where $\chi_{ \{1,3\} }$ denotes the indicator function of the set $\{1,3\}$.
 This implies that the space of fixed points $\mathcal F_{\mathcal L_0}$ is given by $\operatorname{span} \{ \diagstate{0}, \diagstate{2}, \ketbra{}{0}{2}, \ketbra{}{2}{0}\}$. Since $\mathcal L_0$ has gap equal to 1, theorem~\ref{thm:commuting-liouvillians} implies that $\mathcal L_{N}$ satisfies rapid mixing. $\mathcal L_{N}$ forms a uniform family, but it does not satisfy the unique fixed point condition.

  Consider now the following additional Lindbladian
  \[ \mathcal E_0(\rho) = \frac 2 {N} \left[ \ketbra{}{0}{2} \rho \ketbra{}{2}{0} - \frac 12 \{ \rho, \diagstate{2} \} \right] .\]
  Then, we have:
  \[ (\mathcal L_0 + \mathcal E_0)( \ketbra{}{i}{j} )=
  \begin{cases}
    0 & \text{if } i=j=0\\
    \diagstate{0} + \diagstate{2} - 2 \ketbra{}{i}{j} & \text{if } i=j=1,3 \\
    \frac 2{N} (\diagstate{0} - \ketbra{}{i}{j} )& \text{if } i=j=2 \\
    -\left(\chi_{ \{1,3\} } (i) + \chi_{ \{1,3\} }(j) + \frac{ \chi_{\{i,j\}}(2)}{N} \right) \ketbra{}{i}{j} & \text{if } i\neq j .
  \end{cases}\]

Again, this implies that $\mathcal F_{\mathcal L_0 + \mathcal E_0} = \{ \diagstate{0} \}$.
Consequently  $\mathcal L_N + \mathcal E_N$ has a unique fixed point. It is not a uniform family, and it does not satisfy rapid mixing, as it is not even globally gapped. To see this, note that for $\sigma = \diagstate{200\dots0}-\diagstate{020\dots0}$:
\[(\mathcal L_N + \mathcal E_N)(\sigma) = -\frac2 N \, \sigma.\]
Analogously, $\mathcal L_N + \du{\mathcal E_N}$ satisfies the same conditions as $\mathcal L_N + \mathcal E_N$, but the unique fixed point is now $\diagstate{2\ldots 2}$.

All three systems described above are unstable, since we can transform one into the other by applying a perturbation of order $O(1/N)$, yet the fixed points of $\mathcal L_N + \mathcal E_N$ and $\mathcal L_N + \du{\mathcal E_N}$ are locally orthogonal (while $\mathcal L_N$ has both of them as fixed points).
\end{example}

\subsection{Relaxations of the rapid mixing condition}
In the case of finite range or exponentially decaying interactions,
the proof of theorem~\ref{thm:stability} still holds if we relax
equation~\eqref{eq:non-ti-grm} by requiring only a polynomial
decay in time, i.e. a bound of the form
\begin{equation}
\label{eq:relaxed-mixing}
 \eta(T_t^{\overline \Lambda}) \le c \abs{\Lambda}^\delta \gamma (t) ,
\end{equation}
if $\gamma (t)$ is a fast enough decaying function,
where the threshold decay rate is
determined by system-size independent costants (such as the Lieb-Robinson bound
constants and the geometrical dimension of the underlying lattice structure).

Determining the precise value of such threshold requires an argument similar
to the one given for the case of power-law decaying interactions in
section~\ref{sec:power-law-decay}, and is presented in section~\ref{sec:relaxing-rapid-mixing}.


\section{Toolbox for the proof}
\label{sec:toolbox}

Before presenting the proof of theorem~\ref{thm:stability}, we need to introduce some useful tools. We present them in full generality, including the case of power-law decay of interactions, without restricting here to exponentially decaying interactions.

\subsection{Lieb-Robinson bounds for Lindbladian evolution}
\label{sec:lieb-robinson-bounds}
We first recall a generalization of Lieb-Robinson bounds to non-Hamiltonian evolution, due to \cite{Poulin10} and \cite{Nachtergaele12}, which we use to derive a number of useful tools that allow us to approximate the support of an evolving observable with a finite set which grows linearly in time. The following condition is sufficient for the bounds to hold.

\begin{assumption}[Lieb-Robinson condition]
\label{lieb-rob-assumption}
Let $\mathcal L = \sum_{u,r} \mathcal L_{u,r}$ be a local Lindbladian. There exist positive constants $\mu$ and $v$, such that:
\begin{equation}
\label{eq:lieb-robinson-condition}
\sup_{x \in \Lambda} \sum_{u \in \Lambda} \sum_{r \ge \dist(u,x)}
\norm{\mathcal L_{u,r}}_{1\to 1,cb} \abs{b_u(r)} \nu_\mu(r)  \le \frac{v}{2} < \infty;
\end{equation}
where $\nu_\mu(\cdot)$ is one of the following:
\begin{align}
\tag{LR-1}
\label{eq:lieb-robinson-condition-exp}
\nu_\mu(r) &= e^{\mu r}, \\
\tag{LR-2}
\label{eq:lieb-robinson-condition-power}
\nu_\mu(r) &= (1+r)^\mu .
\end{align}
Note that both functions are submultiplicative, in the sense that $\nu_\mu(r+s) \le \nu_\mu(r) \nu_\mu(s)$. Moreover, $\nu_a(r)^b = \nu_{ab}(r)$.

The constant $v$ is called the \emph{Lieb-Robinson velocity} of $\mathcal L$, while the reciprocal function $\nu_\mu^{-1}(r) = 1/\nu_\mu(r)$ is called the \emph{Lieb-Robinson decay} of $\mathcal L$.
\end{assumption}

Note that if $\mathcal L$ has interaction strength $(J, f)$, then condition~\eqref{eq:lieb-robinson-condition} can be replaced by:
\begin{equation}
\label{eq:weak-lieb-robinson-condition}
J \sup_{x,y \in \Lambda} \sum_{n \ge 0} \abs{b_x(r)\setminus b_x(r-1)} \sum_{r \ge n} f(r) \abs{b_y(r)} \nu_\mu(r) \le \frac v {2} < \infty.
\end{equation}
Since our systems are embedded in the lattice $\Z^D$, we have that $v < \infty$, as long as:
\begin{equation}
\sum_{n \ge 0} n^{D-1} F_\mu(n) < \infty, \quad F_\mu(n) := \sum_{r \ge n} r^D \, f(r)\, \nu_\mu(r).
\end{equation}

\begin{remark}
Condition~\eqref{eq:lieb-robinson-condition-exp} is satisfied when $\mathcal L$ has finite-range or exponentially decaying interactions, while condition~\eqref{eq:lieb-robinson-condition-power} is satisfied when $\mathcal L$ has quasi-local interactions. If $\mathcal L$ has interactions decaying as a power-law with exponent $\alpha$, then condition~\eqref{eq:lieb-robinson-condition-power} is satisfied whenever $\alpha > 2D+1$ (by choosing $\mu < \alpha - (2D+1)$).
\end{remark}

\begin{theorem}[Lieb-Robinson bound]
Suppose $\mathcal L$ is a local Lindbladian satisfying assumption~\ref{lieb-rob-assumption}. Let $O_X$ be an observable supported on $X \subset \Lambda$, and denote by $O_X(t) = \du{T_t}(O_X)$ its evolution under $\mathcal L$. Let $K : \mathcal A_Y \to \mathcal A_Y$ be a super-operator supported on $Y \subset \Lambda$ which vanishes on $\identity$. Then, the following bound holds \cite{Poulin10,Nachtergaele12}:
\begin{equation}
\label{eq:lieb-robinson-bound}
\norm{K(O(t))} \le \norm{K}_{\infty \to \infty, cb} \norm{O_X} C(X,Y) \frac{\left( e^{vt} -1 \right)}{\nu_\mu(\dist(X,Y))},
\end{equation}
where $C(X,Y) = \min( \abs{X}, \abs{Y})$.
\end{theorem}

From now on, we will only consider Lindbladians which satisfy equation~\eqref{eq:weak-lieb-robinson-condition} with either of the two possible assumptions on $\nu_\mu(\cdot)$.

\begin{lemma}[Comparing different dynamics]
  \label{lemma:lieb-robinson-comparison}
  Let $\mathcal L_1$ and $\mathcal L_2$ be two local Lindbladians,
  and suppose $\mathcal L_2$ has Lieb-Robinson speed and decay bounded by $v$ and $\nu_\mu^{-1}$.
  Consider an operator $O_X$ supported on $X \subset \Lambda$, and denote by $O_i(t)$ its evolution under $\mathcal L_i,\, i=1,2$.
  Suppose that $\mathcal L_1 - \mathcal L_2 = \sum_{r\ge 0} M_r$,
  where $M_r$ is a superoperator supported on $Y_r$ which vanishes on $\identity$, and $\dist(X, Y_r) \ge r$.
  Then the following holds:
  \begin{equation}
    \norm{O_1(t) - O_2(t)} \le  \norm{O_X} \abs{X} \frac{e^{vt} -vt -1}{v} \sum_{r=0}^\infty  \norm{M_r}_{1\to 1,cb} \nu^{-1}_\mu(r) .
  \end{equation}
\end{lemma}
\begin{proof}
  Let $h(t) = O_1(t) - O_2(t)$. Calculating its derivative, we obtain
  \[ h^\prime(t) = \du{\mathcal L_1} O_1(t) - \du{\mathcal L_2} O_2(t) =
  \du{\mathcal L_1} h(t) + (\du{\mathcal L_1} - \du{\mathcal L_2})O_2(t). \]
  Since $h(0) = 0$, this differential equation for $h(t)$ has solution
  \begin{equation}\label{eq:integral-rep}
    \begin{split}
    h(t) = O_1(t) - O_2(t)
      &= \int_0^t e^{(t-s) \du{\mathcal L_1}} (\du{\mathcal L_1} - \du{\mathcal L_2}) O_2(s) \de s \\
      &= \sum_{r\ge 0} \int_0^t e^{(t-s) \du{\mathcal L_1}} \du{M_r} O_2(s) \de s ,
    \end{split}
  \end{equation}
  giving us a useful integral representation for $O_1(t) - O_2(t)$. From this, we obtain the estimate
  \[ \norm{O_1(t)-O_2(t)} \le  \sum_{r\ge 0} \int_0^t \norm{ \du{M_r} O_2(s) } \de s ,\]
  where we have used the fact that $e^{t \du{\mathcal L_1}}$ is a contraction with respect to $\norm{\cdot}_\infty$ for each $t \ge 0$.

  We can now apply the Lieb-Robinson bound (equation~\eqref{eq:lieb-robinson-bound})
  to each of the terms in the sum in the previous estimate, to obtain:
  \begin{multline*}
    \norm{O_1(t)-O_2(t)} \\
    \le \sum_{r\ge 0} \norm{M_r}_{1\to 1,cb} \norm{O_X} C(X,Y_r) \nu_\mu^{-1}(\dist(X,Y_r)) \int_0^t (e^{vs} -1) \de s ,
  \end{multline*}
  which implies the claimed bound.
  \qed
\end{proof}

A particular application of the previous lemma is when $\mathcal L_2$ is a restriction of $\mathcal L_1$ onto a smaller region.
Since this case occurs frequently, and is of particular interest, we state it as a separate lemma:

\begin{lemma}[Localizing the evolution]
  \label{lem:localizing-evolution}
  Let $O_A$ be an observable supported on a finite $A \subset \Lambda$. Denote by $O_A(t) = \du{T_t}(O_A)$ its evolution under a local Lindbladian $\mathcal L$ with strength $(J, f)$. Given $r >0$, denote by $O_{A}(r;t)$ its evolution under the localized Lindbladian $\mathcal L_{A(r)}$.

Then, the following bound holds:
  \begin{equation}
    \label{eq:localizing-evolution}
   \norm{O_A(t) - O_{A}(r;t)} \le \norm{O_A} \abs{A} J \frac{e^{vt} -1 -vt}{v} \nu^{-1}_\beta(r),
  \end{equation}
  where $\nu^{-1}_\beta(r)$ decays exponentially if $\mathcal L$ satisfies condition~\eqref{eq:lieb-robinson-condition-exp}, while decays as $(1+r)^{-\beta}$ if $\mathcal L$ satisfies condition~\eqref{eq:lieb-robinson-condition-power}. In this case, if we denote by $\alpha$ the decay rate of $\mathcal L$, then $\beta$ is given by:
\[ \beta = \begin{cases}
  \alpha - 3D & \text{if } \alpha \ge 5D -1 ;\\
  \frac{1}{2} (\alpha - D - 1) & \text{if } \alpha \le 5D -1 .
\end{cases} \]
\end{lemma}

\begin{proof}
  First, let us decompose $\mathcal L - \mathcal L_{A(r)}$ as a telescoping sum
  \[ \mathcal L - \mathcal L_{A(r)} = \sum_{l \ge r} \mathcal L_{A(l+1)} - \mathcal L_{A(l)} .\]

  Since each element in the sum is the difference between restrictions on different subsets of the same global Lindbladian,
  it is easy to explicitly write their difference
  \[ \mathcal L_{A(l+1)} - \mathcal L_{A(l)} = \sum_{\delta = 0}^{l+1} \sum_{\dist(u,A) = \delta} \mathcal L_u(l+1 - \delta) .\]

 We group the terms in the sum by their distance from $A$:
 Let
  \[ d = \dist(A, b_u(l+1-\delta)) = \max \{ 0, 2\delta -l -1\} \]
 and
  \begin{align}
    M_0 &= \sum_{l \ge r} \sum_{\delta = 0}^{\frac{l+1}2} \sum_{\dist(u,A) = \delta} \mathcal L_u(l+1-\delta) ;\\
    M_d &= \sum_{l\ge r} \sum_{ \substack{\dist(u,A) = \delta \\ \delta = \frac{l+1+d}{2}}} \mathcal L_u(l+1-\delta) .
  \end{align}
Then, we can write:
 \[ \sum_{d \ge 0} M_d = \mathcal L - \mathcal L_{A(r)} ; \qquad \dist(A, \supp M_d) = d .\]

Applying lemma~\ref{lemma:lieb-robinson-comparison}, we obtain:
  \[ \norm{O_A(t) - O_{A}(r;t)} \le
  \norm{O_A} \abs{A} J \frac{e^{vt} -1-vt}{v} \zeta(r); \]
  where, by denoting $q(l) = \abs{A(l)\setminus A(l-1)}$, $\zeta(r)$ is the following:
  \begin{multline}
  \zeta(r) =
  \frac{1}{J} \sum_{d \ge 0} \norm{M_d}_{1 \to 1, cb} \nu_\mu^{-1}(d) \le \\
   \sum_{l \ge r} \sum_{\delta = 0}^{\frac{l+1}{2}} q(\delta) f(l+1-\delta) +
   \sum_{\delta = \frac{l+1}{2}}^{l+1} q(\delta) f(l+1-\delta) \nu_\mu^{-1}(2\delta - l -1) .
 \end{multline}

  If $\delta \ge (l+1)/2$, since $\nu_\mu(\cdot)$ is submultiplicative, we have:
  \[ \nu_\mu(\delta) \le \nu_\mu( l + 1 -\delta)\nu_\mu(2\delta - l -1).\]
  Otherwise, since $\nu_\mu(\cdot)$ is increasing, we have that $\nu_\mu(\delta) \le \nu_\mu(l+1-\delta)$. Plugging these inequalities in the above sum, we get:
  \[ \zeta(r) \le
  \sum_{l\ge r} \sum_{\delta = 0}^{l+1} \left[ q(\delta) \nu_\mu^{-1} (\delta) \right] \left[ f(l+1 - \delta) \nu_\mu(l+1-\delta) \right] .\]

  Since $f$ satisfies equation~\eqref{eq:weak-lieb-robinson-condition}, which in particular implies
  \[ \sum_{\delta \ge 0} f(\delta) \nu_\mu(\delta) \abs{b_0(\delta)} < \infty ,\]
  then the sequence $f(\delta)\nu_\mu(\delta)$ is decreasing.
  We distinguish two cases:
  If $\nu_\mu$ is of the type~\eqref{eq:lieb-robinson-condition-exp}, then the decay of $f(\delta)\nu_\mu(\delta)$ is exponential.
  Since $q(\delta)$ grows polynomially, $q(\delta) \nu_\mu^{-1}(\delta)$ is exponentially decaying.
  Then, the convolution of the two sequences, which is exactly:
  \[ \sum_{\delta = 0}^{l+1} \left[ q(\delta) \nu_\mu^{-1} (\delta) \right] \left[ f(l+1 - \delta) \nu_\mu(l+1-\delta) \right] \]
  is exponentially decaying too, which implies an exponential decay rate for $\zeta(r)$.
  Thus, there exists some $\beta >0$ such that $\zeta(r) \le \nu_\beta^{-1}(r)$, and this concludes the proof for the case of exponential decay.
  Let us suppose now that $\nu_\mu$ is of type~\eqref{eq:lieb-robinson-condition-power}.
  Then, $f(\delta)\nu_\mu(\delta)$ decays as $(1+\delta)^{\mu-\alpha}$, while $q(\delta) \nu_\mu^{-1}(\delta)$ decays as $(1+\delta)^{D-1-\mu}$.
  This implies\footnote{
    Consider two positive decreasing sequences $(x_n)$ and $(y_n)$.
    Since $0 < p < 1$ implies that $(x+y)^p \le x^p + y^p$, it holds that
   $(x * y)_n^p \le \sum_k x_k^p y_{n-k}^p = (x^p * y^p)_n$.} that their convolution decays as $(1+l)^{-\min(\alpha-\mu, \mu-D+1)}$
 and thus
  \[ \zeta(r) \le c (1+r)^{-\min(\alpha-\mu-1,\mu-D)} = \nu_\beta^{-1}(r) .\]
  Recalling that condition~\eqref{eq:lieb-robinson-condition-power} requires $\mu < \alpha - (2D+1)$, a simple calculation shows that the above decay rate is maximized for
  \[ \mu < \min\left(\alpha-2D-1, \frac{\alpha+D-1}{2}\right),\]
  which gives the claimed formula for $\beta$.
  \qed
\end{proof}

Another specialization of lemma~\ref{lemma:lieb-robinson-comparison}, similar in spirit to the one just presented,
is when we compare the evolution of local observables under $\mathcal L^{A(r)}$ and $\mathcal L^{\overline{A(r)}}$,
as defined in definition~\ref{def:uniform-family}.

\begin{lemma}
\label{lemma:localizing-boundary}
Let $O_A$ be an observable supported on $A\subset \Lambda$. Given $r >0$, it holds that
\begin{equation}
\norm{ {T^*_t}^{\overline{A(r)}}(O_A) - {T^*_t}^{A(r)}(O_A)} \le \norm{O_A}\abs{A}\, J \,\frac{e^{vt} - 1 -vt}{v} \nu_{\beta}^{-1}(r).
\end{equation}

\end{lemma}
\begin{proof}
Without loss of generality, we consider the case of $A(r)$ being a convex set.
By construction, $\mathcal L^{\overline{A(r)}} - \mathcal L^{A(r)} = \mathcal L^{\partial A(r)}$, and
$\mathcal L^{\partial A(r)} = \sum_{d\ge 1} N_d$, where each $N_d$ acts on sites that are closer than $d$ to the border of $A(r)$. We group these terms by their distance from $A$. Let $k=\frac 12 \diam A$ and set:
\begin{align*}
M_0 &= \sum_{i=0}^{k} N_{r+1+i},\\
M_j &= N_{r+1-j}, \quad j = 1 \dots r.
\end{align*}
It is easy to see that $\dist(A, \supp M_j) = j$.
By applying lemma~\ref{lemma:lieb-robinson-comparison}, we have that:
\[ \norm{ {T^*_t}^{\overline{A(r)}}(O_A) - {T^*_t}^{A(r)}(O_A)} \le \norm{O_A}\abs{A} \frac{e^{vt} - 1 -vt}{v} \sum_{j=0}^{r} \norm{M_j}_{1 \to 1, c.b.} \nu_\mu^{-1}(j) .\]
We are left to prove that the sum appearing on the r.h.s. is fast-decaying in~$r$. From definition~\ref{def:uniform-family}
it follows that for $j > 0$:
\[ \norm{M_j}_{1\to 1,c.b.} \le J \abs{\partial_{r-j} A(r)} f(r+1-j) = J \abs{A(r)\setminus A(j)} f(r+1-j) ,\]
while for $j = 0$:
\[ \norm{M_0}_{1\to 1,c.b.} \le \sum_{i=0}^{k} J \abs{\partial_{r+i} A(r)} f(r+1+i) .\]
Setting $h_{m,n} = \abs{ b_0(m) \setminus b_0(n) }$, we have that:
\begin{equation}
\sum_{j=0}^{r} \norm{M_j}_{1 \to 1, c.b.} \nu_\mu^{-1}(j) \le J \zeta(r),
\end{equation}
where $$\zeta(r) := \sum_{i=0}^k h_{r+k, k-i} f(r+1+i) + \sum_{j=1}^r h_{r+k, k+j} f(r+1-j) \nu_\mu^{-1}(j).$$

An argument similar to the one in the proof of lemma~\ref{lem:localizing-evolution} shows that $\zeta(r)$ is fast-decaying. Indeed, $h_{r+k, k-i} f(r+1+i)$ scales asymptotically as $r^D f(r)$, while $h_{r+k, k+j} f(r+1-j)$ scales as $(r-j)^D f(r+1-j)$. If $\mathcal L$ satisfies~\eqref{eq:lieb-robinson-condition-exp}, then $\zeta(r)$ will be exponentially decaying, with rate $\min(\alpha, \mu) - 1 = \mu - 1$.

If otherwise $\mathcal L$ satisfies~\eqref{eq:lieb-robinson-condition-power}, then $\zeta(r)$ has a polynomial decay,
with rate $\min(\alpha - D, \mu) - 1 = \mu -1$. In both cases, then:
\[ \zeta(r) \le \nu_{\mu-1}^{-1}(r) .\]
Notice that the constant $\beta$ defined in lemma~\ref{lem:localizing-evolution} is smaller than $\mu -1$.
\qed
\end{proof}

\subsection{Local rapid mixing}
The rapid mixing condition implies a local version of mixing that will be a useful tool for the proof of theorem~\ref{thm:stability}. We state its definition here.
\begin{definition}[Local rapid mixing]
  \label{def:local-rapid-mixing}
Take $A \subset \Lambda$, and define the \emph{contraction of $T_t$ relative to $A$} as
\begin{align}
    \eta^A(T_t) &:= \sup_{\substack{\rho \ge 0 \\ \trace \rho\, =1 }}
    \norm{\trace_{A^c} \left[ T_t(\rho) - T_{\phi,t}(\rho) \right]}_1 \nonumber \\
    =&\sup_{\substack{\rho \ge 0 \\ \trace \rho\, =1 }} \sup_{\substack{O_A \in \mathcal A_A \\ \norm{O_A}=1}}
    \trace\left(O_A \left[ T_t(\rho) - T_{\phi,t}(\rho) \right]\right)  \nonumber \\
    =&\sup_{\substack{\rho \ge 0 \\ \trace \rho\, =1 }} \sup_{\substack{O_A \in \mathcal A_A \\ \norm{O_A}=1}}
    \trace\left(\rho \left[ T^*_t(O_A) - T^*_{\phi,t}(O_A) \right]\right).
\end{align}
We say that $\mathcal L$ satisfies \emph{local rapid mixing} if, for each $A \subset \Lambda$, we have that
  \begin{equation}
    \label{eq:local-rapid-mixing}
    \eta^{A}(T_t) \le k(\abs{A}) e^{-\gamma t},
  \end{equation}
where $k(r)$ grows polynomially in $r$, $\gamma >0$ and all the constants appearing above are independent of the system size.
\end{definition}

\begin{remark}
\label{obs:local-vs-global-rm}
It follows from the definition that $\eta^A(T_t) \le \eta^B(T_t)$ whenever $A \subset B$.
 In particular, $\eta^A(T_t) \le \eta(T_t)$.
\end{remark}

Note that, in contrast with definition~\ref{def:global-rapid-mixing}, the quantity $\eta^A(T_t)$ depends on the evolution \emph{on the whole system $\Lambda$}, and not just on the subset $A$. Thus local rapid mixing is a very strong condition: the term $k(r)$ appearing in equation~\eqref{eq:local-rapid-mixing} only depends on the support of $A$, so the local mixing time (i.e.\ the time it takes for the reduced density matrix on the subset $A$ to converge) is required to be independent of system size.

\begin{example}
  A simple dissipative system satisfying definition~\ref{def:local-rapid-mixing} is the tensor product of amplitude damping channels acting (with the same rate) on different qubits. Note that, though it might seem a trivial example, there are interesting dissipative systems of this form: among others, dissipative preparation of graph states \cite{Kastoryano12} can be brought into this form by a non-local unitary rotation (which of course does not change the convergence rates).
\end{example}


\section{Proof of main result}
\label{sec:proof}
We are now ready to prove our main resul, theorem~\ref{thm:stability}. The proof proceeds in three steps. First, we show that the assumptions of theorem~\ref{thm:stability} imply that the fixed points of $\mathcal L^{\overline \Lambda}$ for different $\Lambda$ are locally indistinguishable. Then, we prove that rapid mixing implies local rapid mixing. Finally, we show how local rapid mixing and the uniqueness of the initial fixed point imply the desired stability result.

\subsection{Step 1: closeness of the fixed points}
\label{sec:step-1}
Topological quantum order (TQO), namely the property of certain orthogonal quantum states to be locally indistinguishable from each other, is a widely studied property of ground state subspaces in the Hamiltonian setting. In the dissipative setting on the other hand, where the concept of ground states is no longer applicable, one may define the analogous concept for periodic states of Lindbladians. Below we describe the concept of Local Topological Quantum Order (LTQO)~\cite{Spiros11}, which extends the concept of TQO to the invariant subspace (periodic states) of local restrictions of the global Lindbladian.

We note that, in contrast to the Hamiltonian case, in order to prove the
desired stability result we do not require extra assumptions like LTQO, or
frustration-freeness. Indeed, we show in this section that rapid mixing \emph{implies} LTQO and a property similar to frustration-freeness. These properties will play a role in the proof of stability, via lemma~\ref{lemma:frustrated-ltqo-2}.

\begin{definition}[Local Topological Quantum Order (LTQO)]\label{def:LTQO}
Consider a Lindbladian $\mathcal L$.
Take a convex set $A \subset \Lambda$ and let $A(\ell) = \{ x \in \Lambda | \dist (x,A) \le \ell \}$. Given two states $\rho_i\in \mathcal X_{\mathcal L_{A(\ell)}}$, $i=1,2$, consider their reduced density matrices on $A$:
  \[ \rho_i^A = \trace_{A(\ell)\setminus A} \rho_i , \quad i=1,2 . \]
  We say that $\mathcal L$ has \emph{local topological quantum order} (LTQO) if for each $\ell \ge 0$:
  \begin{equation}
    \label{eq:ltqo}
    \norm{\rho_1^A - \rho_2^A}_1 \le p(\abs{A})\, \Delta_0 (\ell),
  \end{equation}
  where $\Delta_0(\ell)$ is a fast-decaying function, and $p(\cdot)$ is a polynomial.
\end{definition}

As a first step in the proof, we will show that the conditions of theorem~\ref{thm:stability} imply that
the fixed point of $T_t$, the fixed point of $T^{\overline A}_t$ and the periodic points of $T^A_t$ are difficult to distinguish locally,
in the same spirit as the LTQO condition.

\begin{lemma}
\label{lemma:frustrated-ltqo}
Let $ \mathcal L$ be a uniform family satisfying condition~\eqref{eq:non-ti-grm}, and suppose each $\mathcal L^{\overline A}$ has a unique fixed point and no other periodic points. Let $O_A$ be an observable supported on $A\subset \Lambda$, $\rho$ a periodic point of $T_t^{A(s)}$ and $\rho^s_\infty$ the unique fixed point of $T_t^{\overline{A(s)}}$. Then, we have
\begin{equation} \label{eq:frustrated-ltqo}
\abs{ \trace O_A (\rho - \rho^s_\infty)} \le \norm{O_A} \left(\frac{J}{v}\abs{A} + c \abs{A}^\delta \right) \Delta_0(s),
\end{equation}
where $$\Delta_0(s) = (\abs{A(s)}/\abs{A})^{\delta v/(v+\gamma)} \nu_{\beta'}^{-1}(s), \quad \beta' = \beta\gamma/(v+\gamma),$$
with $c, \gamma,\delta$ the constants defined in the rapid mixing condition $RM(\gamma,\delta)$, $\beta$ the rate defined in lemma~\ref{lemma:localizing-boundary} and $v$ the Lieb-Robinson velocity.
\end{lemma}
\begin{proof}
Fix a $t := t(s) \ge 0$, to be determined later. Since $T_t^{A(s)}$ acts on its space of periodic points as a unitary evolution, there exists a periodic point of $\mathcal L^{A(s)}$, $\rho^\prime$, such that $\rho = T_t^{A(s)}(\rho^\prime)$.
Then, by the triangle inequality, we have:
\begin{equation} \abs{ \trace O_A (\rho - \rho^s_\infty) }\le
\abs{ \trace O_A [T_t^{A(s)} - T_t^{\overline{A(s)}}](\rho^\prime)) } +
\abs{ \trace O_A ( T_t^{\overline{A(s)}}(\rho^\prime) - \rho^s_\infty)}.
\end{equation}
The first term is bounded by lemma~\ref{lemma:localizing-boundary}, since $$\trace O_A ( T_t^{A(s)}(\rho^\prime) - T_t^{\overline{A(s)}}(\rho^\prime)) = \trace \rho^\prime ( {T^*_t}^{A(s)}(O_A) - {T^*_t}^{\overline{A(s)}}(O_A))$$ and $$\abs{ \trace \rho^\prime ( {T^*_t}^{A(s)}(O_A) - {T^*_t}^{\overline{A(s)}}(O_A)) } \le \|\rho^\prime\|_1 \|{T^*_t}^{A(s)}(O_A) - {T^*_t}^{\overline{A(s)}}(O_A)\|_\infty.$$ The second term is bounded using the rapid mixing condition on $T_t^{\overline{A(s)}}$. By putting the two bounds together, we obtain
  \[ \abs{ \trace O_A (\rho - \rho^s_\infty) } \le \norm{O_A} \abs{A} \frac{J}{v} e^{vt}\nu_\beta^{-1}(s) + \norm{O_A} c \abs{A(s)}^\delta e^{-\gamma t} .\]
Setting $p(s) = (\abs{A(s)}/\abs{A})^\delta$ and choosing $t(s)$ such that
\[ e^{vt(s)} \nu_\beta^{-1}(s) = p(s) e^{-\gamma t(s)},\]
 we have that $t(s) = \ln (\nu_\beta(s) \cdot p(s))^{1/(v+\gamma)}$.
Under such choice, it holds that
\[ e^{-\gamma t(s)} = (\nu_{\beta}(s) p(s))^{-\gamma/(v+\gamma)} =
\nu^{-1}_{\beta^\prime}(s) p(s)^{-\gamma/(v+\gamma)} ,\]
where $\beta' = \beta\gamma/(v+\gamma)$.
Defining
\[\Delta_0(s) := (\abs{A(s)}/\abs{A})^{\delta v/(v+\gamma)} \nu_{\beta'}^{-1}(s),\]
concludes the proof.
\qed
\end{proof}

\begin{corollary}[LTQO]
\label{thm:lto-from-unique}
Under the assumptions of lemma~\ref{lemma:frustrated-ltqo}, the Lindbladian $\mathcal L^{\overline \Lambda}$ satisfies LTQO (definition~\ref{def:LTQO}) for all $\Lambda$.
\end{corollary}
\begin{proof}
Take $A \subset \Lambda$, and $s \ge 0$. Let $\rho_1$ and $\rho_2$ be two periodic points of $T_t^{A(s)}$. Then, by the triangle inequality, we have that:
\begin{multline*}
 \abs{ \trace O_A (\rho_1 - \rho_2)} \le \abs{ \trace O_A (\rho_1 - \rho^s_\infty)} + \abs{ \trace O_A (\rho^s_\infty - \rho_2) } \le \\
\le 2 \norm{O_A} \left(\frac{J}{v}\abs{A} + c \abs{A}^\delta \right) \Delta_0(s).
\end{multline*}
Since $\|\rho^A_1-\rho^A_2\|_1 = \sup_{\|O_A\|=1} \abs{ \trace O_A (\rho_1 - \rho_2)}$, the result follows immediately.
\qed
\end{proof}

\begin{lemma}
  \label{lemma:frustrated-ltqo-2}
Under the same notation and assumptions of lemma~\ref{lemma:frustrated-ltqo}, we have the following bound for $\rho_\infty$ the unique fixed point of $T_t$:
\begin{equation}
\sup_{\|O_A\|=1}  \abs{ \trace O_A ( \rho_\infty - \rho_\infty^s ) } \le
 \norm{O_A} \left(\frac{J}{v}\abs{A} + c \abs{A}^\delta \right) \Delta_0(s).
\end{equation}
\end{lemma}
\begin{proof}
By the triangle inequality:
\begin{align*}
    \abs{ \trace O_A ( \rho_\infty - \rho_\infty^s) )} \le \abs{ \trace O_A ( \rho_\infty - T_t^{\overline{A(s)}}(\rho_\infty) }
    + \abs{ \trace O_A ( T_t^{\overline{A(s)}}(\rho_\infty) - \rho_\infty^s)}.
\end{align*}
The first term on the right can be bounded using lemmas~\ref{lem:localizing-evolution} and~\ref{lemma:localizing-boundary} along with $T_t(\rho_\infty)=\rho_\infty$:
\begin{eqnarray*}
&\abs{ \trace O_A \big ( T_t(\rho_\infty) - T_t^{\overline{A(s)}}(\rho_\infty) \big) } = \abs{ \trace \rho_\infty (T^*_t(O_A) - {T^*_t}^{\overline{A(s)}}(O_A))} \\
&\le \|\rho_\infty\|_1 \left(\|T^*_t(O_A) - {T^*_t}^{A(s)}(O_A)\|_\infty+\|{T^*_t}^{A(s)}(O_A) - {T^*_t}^{\overline{A(s)}}(O_A)\|_\infty\right) \\ &\le \norm{O_A}\abs{A} \frac{J}{v} e^{vt} \nu_\beta^{-1}(s).
\end{eqnarray*}
The second term is bounded using the rapid mixing condition:
\[ \abs{ \trace O_A \big ( T_t^{\overline{A(s)}}(\rho_\infty) - \rho_\infty^s \big)} \le \norm{O_A} c \abs{A}^\delta p(s) e^{-\gamma t}.\]
By making the same choice of $t=t(s)$ as in lemma~\ref{lemma:frustrated-ltqo}, we get the desired bound.
\qed
\end{proof}
\begin{corollary}[Approximate frustration-freeness]
\label{cor:approx-ff}
Under the same notation and assumptions of lemma~\ref{lemma:frustrated-ltqo}, denote by $\rho_\infty$ the unique fixed point of $T_t$, and by $\rho$ a periodic point of $T_t^{A(s)}$. Then, we have the following bound:
\begin{equation}\label{eq:approx-ff}
\sup_{\|O_A\|=1} \abs{ \trace O_A( \rho_\infty - \rho )} \le
2 \norm{O_A} \left(\frac{J}{v}\abs{A} + c \abs{A}^\delta \right) \Delta_0(s).
\end{equation}
\end{corollary}
\begin{proof}
By the triangle inequality and lemmas~\ref{lemma:frustrated-ltqo} and~\ref{lemma:frustrated-ltqo-2}, we have:
\begin{multline*}
\abs{ \trace O_A ( \rho_\infty - \rho) )} \le \abs{ \trace O_A ( \rho_\infty - \rho_\infty^s) } +  \abs{ \trace O_A ( \rho_\infty^s - \rho) } \le \\
\le 2 \norm{O_A} \left(\frac{J}{v}\abs{A} + c \abs{A}^\delta \right) \Delta_0(s).
\end{multline*}
\qed
\end{proof}

\subsection{Step 2: from global to local rapid mixing}
\label{sec:step-2}
As a second step in the proof, we show that the assumptions on $\mathcal L$ imply local rapid mixing.

\begin{proposition}[From global to local rapid mixing]
\label{prop:global-local-mixing}
  Let $\mathcal L$ be a uniform family of Lindbladians with unique fixed point. Then, if condition~\eqref{eq:non-ti-grm} is satisfied, $\mathcal L$ satisfies local rapid mixing.
\end{proposition}
\begin{proof}
  Let $O_A$ be an observable supported on $A$ with $\norm{O_A} =1 $.
  Denote by $s_0$ the minimum $s \ge 0$ such that $A(s) = \Lambda$.
  Fix $0\le s \le s_0$, and let $B=A(s)$.
  Then, by the triangle inequality, we can bound the norm of $(\du{T_t} - \du{T_{\infty}} )$ as follows:
  \begin{equation}
    \label{eq:ulmt}
    \norm{ (\du{T_t} - \du{T_{\infty}} ) O_A}  \le \norm{ ( \du{T_t} - T_t^{\overline B *} ) O_A}
    + \norm{ (T_{t}^{\overline B *}  -  T_{\infty}^{\overline B *} ) O_A} + \norm{ (T_{\infty}^{\overline B *} - T_{\infty}^{*} ) O_A }.
\end{equation}
We bound the first term on the right using lemmas~\ref{lem:localizing-evolution} and~\ref{lemma:localizing-boundary}:
  \begin{equation}
    \label{eq:ulmt-1}
    \norm{( \du{T_t}  - T_t^{\overline B *} ) O_A}  \le \abs{A} \frac{J}{v} (e^{vt} - 1 -vt) e^{-\beta s} .
  \end{equation}
The second term is bounded by the rapid mixing condition~\eqref{eq:non-ti-grm}, setting $p(s) = (\abs{A(s)}/\abs{A})^\delta$:
  \begin{equation}
    \label{eq:ulmt-2}
    \norm{ (T_{t}^{\overline B *}  -  T_{\infty}^{\overline B *} ) O_A} \le \eta(T_t^{\overline B}) \le c \abs{A}^\delta p(s) e^{-\gamma t}.
  \end{equation}
Finally, the third term is bounded by using lemma~\ref{lemma:frustrated-ltqo-2}:
  \begin{equation}
    \label{eq:ulmt-3}
    \norm{ (T_{\infty}^{\overline B *}  -  T_\infty^* ) O_A} = \abs{ \trace O_A (\rho_\infty^s  - \rho_\infty) }
    \le  \left(\frac J v \abs{A} + c \abs{A}^\delta\right) \Delta_0(s).
  \end{equation}
Substituting bounds~\eqref{eq:ulmt-1}, \eqref{eq:ulmt-2} and~\eqref{eq:ulmt-3} into equation~\eqref{eq:ulmt}, we obtain, for
$0 \le s \le s_0$ and for all $t \ge 0$:
  \[ \eta^A(T_t) \le \frac{J}{v} \abs{A} e^{v t} e^{-\beta s} + c \abs{A}^\delta p(s) e^{-\gamma t} +  \left(\frac J v \abs{A} + c \abs{A}^\delta\right) \Delta_0(s) .\]
We want to show that we can choose $s = s(t) \in [0,s_0]$ in such a way that both $e^{vt} e^{-\beta s}$ and $e^{-t \gamma} \, p(s) $ are exponentially decaying in $t$. Choose $s := s(t) = t(v+\gamma)/\beta$. Since $\Delta_0(s) = (\abs{A(s)}/\abs{A})^{\delta v/(v+\gamma)} \nu_{\beta'}^{-1}(s)$, denoting
\[ \bar p(t) = p \circ s (t) = p( t (v+\gamma)/\beta ) ,\]
we have that
\[ \Delta_0(s(t)) = \bar p( t )^{v/(v+\gamma)} e^{-\gamma t}.\]
Therefore, since $\bar p(t) \ge 1$,
\begin{multline*}
\eta^A(T_t) \le \frac{J}{v} \abs{A} e^{-\gamma t} +
c \abs{A}^\delta \bar p(t) e^{-\gamma t} +
\left(\frac J v \abs{A} + c \abs{A}^\delta \right) \bar p(t)^{v/(v+\gamma)} e^{-\gamma t} \le \\
\le 2 \left(\frac J v \abs{A} + c \abs{A}^\delta \right) \bar p(t) e^{-\gamma t},
\quad \forall t \le \frac{\beta}{v+\gamma} s_0.
\end{multline*}
When $t \ge \beta/(v+\gamma) s_0$, we can simply bound $\eta^A(T_t)$ by $\eta(T_t)$ (see remark~\ref{obs:local-vs-global-rm}), obtaining:
\[ \eta^A(T_t) \le c \abs{A}^\delta p(s_0) e^{-\gamma t} \le c \abs{A}^\delta \bar p(t) e^{-\gamma t}, \quad \forall t \ge \frac{\beta}{v+\gamma} s_0 .\]
This completes the proof.
\qed
\end{proof}

\subsection{Step 3: from local rapid mixing to stability}
\label{sec:step-3}
We now prove that local rapid mixing alone implies stability. This is the last step in the proof of theorem~\ref{thm:stability}, as we already proved in the previous sections that the condition of theorem~\ref{thm:stability} imply local rapid mixing. However, the following result also stands independently: if a system can be shown to satisfy local rapid mixing by other means, it will also be stable. Moreover, the same proof holds if we relax the assumption on
prefactor $k(\abs{A})$ in equation~\eqref{eq:local-rapid-mixing}: a similar (but weaker) stability result will
hold true as long as $\abs{A}$ is independent of system size.

\begin{theorem}
  \label{thm:stability-single-fixed-point}
  Let $\mathcal L$ be a local Lindbladian satisfying local rapid mixing, and having a unique fixed point $\rho_\infty$ such that
  \[ T^*_\phi(O_A) = T^*_\infty(O_A) = \trace(O_A \rho_\infty) \identity .\]
  Then, using the notation of theorem~\ref{thm:stability}, for all observables $O_A$ supported on $A \subset \Lambda$ we have that
  \begin{equation}\label{eq:stability2}
    \norm{T_t^{*}(O_A) - S_t^{*}(O_A)} \le c(\abs{A})\, \norm{O_A} \left( \epsilon + \abs{\Lambda} \nu^{-1}_\eta(d_A) \right),
  \end{equation}
  where $d_A = \dist(A, \Lambda^c)$;
  $\eta$ is positive and independent of $\Lambda$;
  $\nu_\eta^{-1}(d) \le (1+d)^{-D-1}$;
  $c(|A|)$ is independent of $\Lambda$ and $t$, and is bounded by a polynomial in $\abs{A}$.
\end{theorem}

\begin{proof}
  Let $O_0(t)=\du T_t(O_A)$ and $O_1(t)=\du S_t(O_A)$ and write the difference $O_0 - O_1$ using the integral representation from equation~\eqref{eq:integral-rep}:
  \[ O_0(t) - O_1(t) = \int_0^t  \du{S_{t-s}} \du{E} \du{T_s}(O_A) \de s.\]
The triangle inequality implies:
  \[ \norm{O_0(t) - O_1(t)} \le \sum_u \sum_r \int_0^t \norm{ \du{E_{u,r}} O_0(s)} \de s +  \sum_d \int_0^t \norm{ \du{E_{d}} O_0(s)} \de s ,\]
  where we used the fact that $S_t$ is a contraction.

  Fix a $K \in \{ E_{u,r}\}_{u,r} \cup \{ E_d \}_d$, and let $\delta = \dist(A,\supp K)$. We can split the integral at a time $t_0$ (to be fixed later, depending on $\delta$). We bound the first part of the integral with Lieb-Robinson bounds:
  \[ \int_0^{t_0} \norm{\du K O_0(s)} \de s \le \norm{K}_{1\to 1,cb} \norm{O_A} \abs{A}\frac{e^{v t_0} -vt_0 -1}{v \nu_\mu(\delta)} . \]
  Now pick $t_0 = t_0(\delta)$ such that
  \[ \nu^{-1}_\mu(\delta) \frac{e^{v t_0} -v t_0 -1}{v} \le \nu^{-1}_{\mu/2}(\delta) .\]
  We can choose $t_0(\delta) = \frac{\mu}{2} \frac{\log v}{v} \delta = O(\delta)$, for exponentially decaying (or faster) $\nu^{-1}_\mu(\delta)$.

  If $t \le t_0(\delta)$, then we have bounded the entire integral, and we are done. Otherwise, we treat the second part of the integral as follows:
  \begin{align*}
    &\int_{t_0(\delta)}^t \norm{\du{K} O_0(s)} \de s
    = \int_{t_0(\delta)}^t \norm{\du{K} (O_0(s) - \du{T_\infty}(O_A) )} \de s\\
    &\le \norm{K}_{1\to 1,cb}  \norm{O_A} \int_{t_0(\delta)}^\infty \eta^A(T_s) \de s
     \le \norm{K}_{1\to 1,cb}  \norm{O_A} q(\abs{A}) \int_{t_0(\delta)}^\infty e^{-\gamma s} \de s\\
    &= \norm{K}_{1\to 1,cb}  \norm{O_A} k(\abs{A}) \frac{1}{\gamma} e^{- \gamma t_0(\delta)}
  \end{align*}
  where we used
  $\du K \du T_\infty(O_A) = \du K( \trace(\rho_\infty O_A) \identity) = \trace(\rho_\infty O_A) \du K (\identity) = 0$,
  together with the local rapid mixing condition.

  Since $t_0(\delta)$ is linear in $\delta$, we have that:
  \[h(\delta) := e^{-\frac{\mu \delta}{2}} + \frac{1}{\gamma} e^{- \gamma t_0(\delta)} \]
  is exponentially decaying in $\delta$.

  Putting the different bounds together, we obtain:
  \[ \int_0^t \norm{\du{K} O_0(s)} \de s \le \norm{K}_{1\to 1,cb}  \norm{O_A} k_1(\abs{A})  h(\delta) ,\]
  where $k_1(\abs{A}) = \max(k(\abs{A}), \abs{A})$.

  Returning to the sum, we have proven that:
\begin{multline}
\label{eq:main-thm-summable}
    \norm{O_0(t) - O_1(t)} \le
    \epsilon \ k_1(\abs{A}) \norm{O_A} \Big[
    \underbrace{ \sum_u \sum_r e(r) h( \dist(A,b_r(u))) }_{I_1(A;\, e,h)}\\
    + \underbrace{\sum_d \abs{\partial_d \Lambda} e(d) h( \dist(A, \partial_d \Lambda) )}_{I_2(A;\, e ,h)}
    \Big] .
\end{multline}
It suffices to show that $I_1$ and $I_2$ are finite (and independent of system size), and that $I_2$ decays exponentially in $\dist(A,\Lambda^c)$.
Let us decompose the $I_1$ as follows
\begin{align*}
    &I_1(A;\, e,h) = \sum_u \sum_r e(r) h( \dist(A,b_r(u)))  \\
    &= \mspace{-20mu} \sum_{\dist(u,A) = 0} \sum_r e(r) h(0) + \sum_{d > 0}\; \sum_{\dist(u,A) = d} \left( \sum_{r =0}^d e(r) h(d-r) + \mspace{-7mu} \sum_{r=d+1}^\infty e(r) h(0) \right ) \\
    &= h(0) \abs{A} \sum_r e(r) + \sum_{d > 0} q(d) \left( \sum_{r=0}^d e(r) h(d-r) + h(0) \sum_{r=d+1}^\infty e(r) \right),
  \end{align*}
  where $q(d) = \abs{ \{ u : \dist(u,A) = d \} }$ grows polynomially in $d$.

The first term is clearly bounded, since $e(r)$ is summable. Since $e$ and $h$ are both exponentially decaying functions, their discrete convolution $e \star h (d) = \sum_{r=0}^d e(r) h(d-r)$ is also exponentially decaying, and consequently summable against any polynomial. The same holds for $\sum_{r>d} e(r)$. This proves that the second term is also bounded.

On the other hand, we have that
\[ I_2(A;\, e,h) =  \sum_d \abs{\partial_d \Lambda} e(d) h( \dist(A, \partial_d \Lambda) ) \le \abs{\Lambda}( e\star h(d_A) + \sum_{d\ge d_A} e(d)) ,\]
where $d_A = \dist(A, \Lambda^c)$.
We have just proven that $e\star h(d_A)$ and $\sum_{d \ge d_A} e(d)$ are exponentially decaying. This implies that there exists a positive $\eta$ such that
$\nu_\eta^{-1} (d_A)$ upper bounds both.
  Denoting $ c(\abs{A}) = k_1(\abs{A}) I_1(A;\, e,h)$, we have the desired bound.
\qed
\end{proof}
\subsection{Power-law decay}
\label{sec:power-law-decay}
As we stated before, the results and proofs presented above still hold when $\mathcal L$ has  quasi-local or power-law interactions. In the latter case, this is only true when certain conditions are met on the decay of $\mathcal L$. In what follows, we highlight the changes one needs to make in the case of power-law decay, in order for the main stability results to hold.

\begin{definition}[Compatibility condition]
\label{def:compatibility}
Let $\mathcal L$ be a local Lindbladian, and suppose it satisfies~\eqref{eq:lieb-robinson-condition-power} and rapid mixing RM($\gamma$, $\delta$).
Let $\mu$ and $v$ be the Lieb-Robinson constants for $\mathcal L$ defined in assumption~\ref{lieb-rob-assumption}
and $\beta$ the constant defined in lemma~\ref{lem:localizing-evolution}.
Then we say that $\mathcal L$ satisfies the weak compatibility condition for stability, if the following inequality is satisfied.
\begin{equation}
\tag{CC-1}
\label{eq:cc-1}
\beta\gamma -  \delta D v > 0;
\end{equation}
we say that $\mathcal L$ satisfies the strong compatibility condition for stability if
\begin{equation}
\tag{CC-2}
\label{eq:cc-2}
\mu \frac{\beta \gamma - \delta Dv}{\beta (\gamma + v)} > D+2.
\end{equation}
Moreover, if the perturbation $E$, defined in theorem~\ref{thm:stability}, is decaying polynomially and not exponentially, it must satisfy
\begin{equation}
\tag{CC-e}
\label{eq:cc-e}
 \sum_n n^D \sum_{r>n} e(r) < \infty
\end{equation}
for the theorem to hold.
\end{definition}

\begin{remark}
Clearly, the strong version of the compatibility condition implies the weak one.
If $\mathcal L$ has quasi-local interactions, then the (polynomial) decay rate $\alpha$ of the interactions can be chosen to be larger than any fixed value.
Consequently, since $\beta$ and $\mu$ can be taken to be linear in $\alpha$, quasi-local Lindbladians $\mathcal L$ satisfy the strong compatibility condition~\eqref{eq:cc-2}.
\end{remark}

Under the weak compatibility condition, all the results presented in sections~\ref{sec:step-1} and~\ref{sec:step-2} still hold true,
while under the strong compatibility condition also the results presented
in~\ref{sec:step-3} are still valid, and in particular our main result, theorem~\ref{thm:stability-single-fixed-point}.

We will now show this in the cases in which we made explicit use of condition~\eqref{eq:lieb-robinson-condition-exp}, and give the needed modifications to the proofs of
lemma~\ref{lemma:frustrated-ltqo}, proposition~\ref{prop:global-local-mixing} and theorem~\ref{thm:stability-single-fixed-point} in order to make them valid for power-law decaying interactions.

From now on, we proceed under the working hypothesis that $\mathcal L$ satisfies~\eqref{eq:lieb-robinson-condition-power} and that the above compatibility conditions are satisfied.

\begin{proof}[Modifications in the proof of lemma~\ref{lemma:frustrated-ltqo}]
The argument below follows closely the proof of the original lemma, but now one must check that $\Delta_0(s)$ is still decaying.
Recall the definition of $\Delta_0(s)$ from the original proof of lemma~\ref{lemma:frustrated-ltqo}:
\[\Delta_0(s) = (\abs{A(s)}/\abs{A})^{\delta v/(v+\gamma)} \nu_{\beta'}^{-1}(s), \quad \beta' = \beta\gamma/(v+\gamma).\]
Since $\left(\abs{A(s)}/\abs{A}\right)^{\delta v/(v+\gamma )}$ grows as $(1+s)^{\delta D v/(v+\gamma )}$,
we have:
\[ \Delta_0(s) \sim (1+s)^{- \gamma'} ,\]
where $\gamma' = \frac{\beta\gamma - \delta D v}{v+\gamma}$ is positive because of~\eqref{eq:cc-1}.
\qed
\end{proof}

\begin{proof}[Modifications in the proof of proposition~\ref{prop:global-local-mixing}]
Keeping the notation introduced in the original proof of this proposition, we have already shown that, for each $0 \le s \le s_0$:
  \[ \eta^A(T_t) \le \frac{J}{v} \abs{A} e^{v t} \nu^{-1}_{\beta}(s)  + c \abs{A}^\delta p(s) e^{-\gamma t}
  + \left(\frac{J}{v} \abs{A} + c \abs{A}^{\delta} \right) \Delta_0(s). \]
At this point, we can no longer choose $s=s(t)$ to scale linearly in $t$, since the decay $\nu^{-1}_{\beta}(s)$ is polynomial in $s$ and the prefactor $e^{vt}$ would render the bound trivial. Still, we may choose $s = s(t) \in [0,s_0]$ in such a way that the r.h.s. above is exponentially decaying in $t$.

Fix $k >0$ (to be determined later), and consider:
\[ s(t) = e^{kt} -1 ,\]
in such a way that for $t \le \log(1+s_0)/k$, we have:
\[ \bar p(t) = p \circ s(t) =\left(\abs{A(e^{kt}-1)}/\abs{A} \right)^\delta \sim e^{kD\delta t} .\]
Then, the r.h.s. of the desired bound for $\eta^A(T_t)$ contains the following exponentials:
\[ e^{vt} \nu^{-1}_{\beta}(s) = e^{-(\beta k-v)t} ; \quad p(s)e^{-\gamma t} \sim e^{-(\gamma - k D \delta)t} ,\]
and
\[ \Delta_0(s) \sim (1+s)^{- \gamma'} = e^{-k\gamma' t} ,\]
where $$\gamma^\prime = \frac{\beta\gamma - \delta D v}{v+\gamma}$$ is defined in the modified proof of lemma~\ref{lemma:frustrated-ltqo}. We want to show that we can choose $k$ in such a way that all the exponential functions appearing above are decaying, i.e. each exponent is negative for $t > 0$. \eqref{eq:cc-1} implies that $\Delta_0(s)$ is decaying for all $k>0$.
Let
\[ k' = \frac{v + \gamma}{\beta + \delta D} ,\]
such that $\beta k' - v = \gamma - k' D \delta = k' \gamma'$, making all of the above exponents equal to $-(\beta\gamma - \delta D v)/(\beta+\delta D)$ and negative (due to~\eqref{eq:cc-1}), as desired.

When $t \ge \log(1+s_0)/{k'}$, as in the proof for exponentially decaying interactions, we bound $\eta^A(T_t)$ by $\eta(T_t)$ (see remark~\ref{obs:local-vs-global-rm}), thus obtaining:
\[ \eta^A(T_t) \le c \abs{A}^\delta  p(s_0) e^{-\gamma t} \le c \abs{A}^\delta \bar p(t) e^{-\gamma t}
  \sim c \abs{A}^\delta e^{- k' \gamma' t}.\]
\qed
\end{proof}

\begin{proof}[Modifications in the proof of theorem~\ref{thm:stability-single-fixed-point}]
  Following the same steps as in the original proof,
  but now using the constants for the local rapid mixing obtained in the modified proof of proposition~\ref{prop:global-local-mixing},
  we have that, for each $0 \le t_0 \le t$:
  \begin{multline*}
    \int_0^t \norm{\du K O_0(s)} \de s \le
  \norm{K}_{1\to 1, cb} \norm{O_A} \left ( \abs{A}\frac{1}{v} e^{v t_0} \nu_\mu^{-1}(d) + k(\abs{A})  e^{-t_0 \frac{\beta \gamma - \delta D v}{\beta + \delta D}} \right),
  \end{multline*}
  where $d = \dist(A,\supp K)$.

  Let us define $t_0(d) = k \log(1+d)$ for some positive $k$ (to be determined later), and denote
  $h(d) =  \nu_{vk - \mu}(d) + \nu_{-k\frac{\beta \gamma - \delta D v}{\beta + \delta D}}(d)$, such that
  \[ \int_0^t \norm{K O_0(s)} \de s \le \norm{K}_{1\to 1,cb}\norm{O_A}k_1(\abs{A}) h(d),\]
  where $k_1(\abs{A}) = \max(k(\abs{A}), \abs{A}/v)$.
  Then we have that $h$ has a maximum decay rate of
  \[ \mu^\prime = \sup_{k\ge 0} \min\left(\mu - vk, k \frac{\beta \gamma - \delta D v}{\beta + \delta D}\right) .\]
  The optimal choice of $k$ is $k  = \frac{\mu}{\beta} \frac{\beta + \delta D}{v+\gamma}$,
  in such a way that
  $\mu^\prime = \frac{\mu}{\beta} \frac{\beta \gamma - \delta Dv}{v+\gamma}$.
  $\mu^\prime $ is positive because of condition~\eqref{eq:cc-1}.

  Recalling the following definitions from the original proof of theorem~\ref{thm:stability-single-fixed-point}:
  \begin{align*}
  &q(d) = \abs{ \{ u : \dist(u,A) = d \} }, &
  l(d) = \abs{\partial_d \Lambda} e(d), \\
  &x \star y(d) = \sum_{r=0}^d x(r) y(d-r), &
  d_A = \dist(A, \Lambda^c),
  \end{align*}
we need to show that
  \[
  I_1(A;\, e,h) =  h(0) \abs{A} \sum_r e(r) + \sum_{d > 0} q(d) \left( e \star h (d) + h(0) \sum_{r>d} e(r) \right)
  e\star h(d)
  \]
is finite, and that
  \[ I_2(A;\, e,h) =   l \star h(d_A) + \sum_{d\ge d_A} l(d) \le \nu_\eta^{-1}(d_A)\]
  for some positive $\eta$.
Notice that
\[ I_1(A;\, e,h) \le (1+\abs{A})h(0) \sum_d q(d) \sum_{r>d} e(r) + \sum_{d} q(d)\, e\star h(d) .\]
  Since $q(d)$ grows as $(1+d)^{D}$, $\sum_d q(d) \sum_{r\ge d} e(r)$ is finite if
$ \sum_n n^D \sum_{r> n} e(r) < \infty$,
which is condition~\eqref{eq:cc-e}.

  On the other hand, $e \star h$ decays as the slowest of the two functions. Since we have already
  assumed that $\sum_{d} q(d) \sum_{r>d} e(r)$ is finite, we only need to satisfy that $\sum_d q(d) \sum_{r > d}h(r)$ is finite.
For this to happen, it suffices that $\mu^\prime > D+2$, which is condition~\eqref{eq:cc-2}.

In order to bound $I_2(A; e,h)$, note that $l(d) \le \abs{\Lambda} e(d)$, and therefore \[I_2(A; e,h)\le \abs{\Lambda} ( e \star h (d_A) + \sum_{d \ge d_A} e(d) ).\] We have already proven that conditions~\eqref{eq:cc-2} and~\eqref{eq:cc-e} imply that the r.h.s. of the latter bound is decaying polynomially in $d_A$ at least as fast as $\nu_{D+1}^{-1}(d_A)$.

 This concludes the proof.
\qed
\end{proof}

\subsection{Relaxing rapid mixing}
In this section, we will show that, in the case of exponentially decaying interactions~\eqref{eq:lieb-robinson-condition-exp},
the proof of theorem~\ref{thm:stability} still holds if
\begin{equation}
\label{eq:cc-relaxed-mix}
\sum_n n^D \sum_{d>n}
\int_d^\infty \left(1 + \frac{v}{\beta}s - \frac{1}{\beta} \log \gamma(s)\right)^{\delta D} \gamma (s) \de s < \infty.
\end{equation}

We will directly prove proposition~\ref{prop:global-local-mixing},
without the intermediate step of a lemma like~\ref{lemma:frustrated-ltqo}. Nonetheless,
results of that kind can be proven using exactly the same arguments that
we will use in the following proof.
\label{sec:relaxing-rapid-mixing}
\begin{proof}[Proof of proposition~\ref{prop:global-local-mixing}]
Using the same notation as in the original proof, we have that for $0 \le s \le s_0$:
\[
    \norm{ (\du{T_t} - \du{T_{\infty}} ) O_A}  \le
    \norm{ ( \du{T_t} - T_t^{\overline B *} ) O_A}
    + \norm{ (T_{t}^{\overline B *}  -  T_{\infty}^{\overline B *} ) O_A}
    + \norm{ (T_{\infty}^{\overline B *} - T_{\infty}^{*} ) O_A }.
\]

We will bound the first two terms as in the original proof
(using lemma~\ref{lem:localizing-evolution}, lemma~\ref{lemma:localizing-boundary} and equation~\eqref{eq:relaxed-mixing})
while we rewrite the third term as in the proof of lemma~\ref{lemma:frustrated-ltqo}:
\begin{multline*}
  \norm{ (T_{\infty}^{\overline B *}  -  T_\infty^* ) O_A} =
  \abs{ \trace O_A (\rho_\infty^s  - \rho_\infty) }
  \le \\
  \abs{ \trace O_A ( \rho_\infty - T_t^{\overline{A(s)}}(\rho_\infty)) }
  + \abs{ \trace O_A ( T_t^{\overline{A(s)}}(\rho_\infty) - \rho_\infty^s)}
  = \\
  \abs{ \trace O_A (T_t(\rho_\infty) - T_t^{\overline{A(s)}}(\rho_\infty)) }
  + \norm{O_A} \norm{ T_t^{\overline{A(s)}}(\rho_\infty) - \rho_\infty^s}_1
  \le \\
  \norm{ ( \du{T_t} - T_t^{\overline B *} ) O_A}
  + \eta(T_t^{\overline B}) .
\end{multline*}
Thus we have that
\[ \norm{ (\du{T_t} - \du{T_{\infty}} ) O_A}
\le 2 \frac{J}{v} \abs{A} e^{v t} \nu^{-1}_{\beta}(s)  + 2 c \abs{A}^\delta p(s) \gamma (t),
\]
where $p(s) = (\abs{(A(s)}/\abs{A})^\delta \sim (1+s)^{\delta D}$.

We have claimed that the result only holds in the case of
exponentially decaying or faster decay of interaction.
Suppose $\nu_\beta(s) = (1+s)^\beta$ (i.e., if $\mathcal L$ satisfies~\eqref{eq:lieb-robinson-condition-power}).
Defining  $s = s(t)$ as
\[ s(t) = e^{\frac{v}{\beta + \delta D} t} \gamma(t)^{-\frac{1}{\beta + \delta D}} -1 ,\]
then it holds that
\[\frac{e^{vt}}{\gamma(t)}  = p(s(t))  \nu_\beta (s(t)) \quad \forall t \le t_0,\]
where $t_0$ is such that $s(t_0) = s_0$.
Thus
\[ \delta_0(t) := e^{vt} \nu_\beta^{-1}(s) = e^{\frac{\delta D}{\beta + \delta D} vt} \gamma(t)^{\frac{\beta}{\beta + \delta D}} .\]
We have that this last function is decaying in $t$ if
\[ \gamma(t) < e^{- (v \delta D/\beta) t} .\]
This forces $\gamma (t)$ to be exponentially decaying,
and thus there is no possible relaxation of the rapid mixing condition.

On the other hand, if $\nu_\beta(s) = e^{\beta s}$ (i.e., if $\mathcal L$ satisfies~\eqref{eq:lieb-robinson-condition-exp}),
we define
\[ s(t) =  \frac{v}{\beta} t - \frac{1}{\beta} \log \gamma (t) ,\]
such that $e^{vt}\nu_\beta^{-1}(s) = \gamma(t)$
and
\[ \bar p(t) = p \circ s(t) \sim \left(1+ \frac{v}{\beta}t - \frac{1}{\beta}\log \gamma(t)\right)^{\delta D} \]
grows polynomially.

In this case, we have proved that
\[ \norm{ (\du{T_t} - \du{T_{\infty}} ) O_A} \le 2 \left( \frac{J}{v} \abs{A} + c \abs{A}^\delta \right) \bar p(t) \gamma(t) ,\]
and this concludes the proof since equation~\eqref{eq:cc-relaxed-mix} implies that
 $\bar p(t) \gamma(t)$ is decaying in $t$.
\qed
\end{proof}

\begin{proof}[Proof of theorem~\ref{thm:stability-single-fixed-point}]
Following the same steps as in the original proof,
we have that for any $K \in \{ E_{u,r}\}_{u,r} \cup \{ E_d \}_d$, and let $\delta = \dist(A,\supp K)$
\[ \int_0^t \norm{\du{K} O_0(s)} \de s \le \norm{K}_{1\to 1,cb}  \norm{O_A} k_1(\abs{A})  h(\delta) .\]
where $ h(\delta)$ is now
\[ h(\delta) = e^{-\mu \delta/2} +  \int_{\frac{\mu}{2} \frac{\log v}{v} \delta}^\infty \bar p(s) \gamma(s) \de s .\]
We want to show that $h(\delta)$ is decaying fast enough for the r.h.s. of equation~\eqref{eq:main-thm-summable} to be summable.
This is the case (see the proof of theorem~\ref{thm:stability-single-fixed-point} in section~\ref{sec:power-law-decay})
if equation~\eqref{eq:cc-relaxed-mix} holds.
\qed
\end{proof}


\section{Glauber dynamics}\label{sec:glauber-dynamics}

\subsection{Quantum embedding of Glauber dynamics}
As an example of a non-trivial dynamics for which we can now prove stability using our results, we turn to one of the most studied dynamics in classical statistical mechanics: Glauber dynamics, a Markov process that samples thermal states of local (classical) Hamiltonians on lattices. Apart from being an interesting model in itself, it has important applications in Monte-Carlo Markov chain algorithms for numerical many-body physics \cite{Liggett85}. Determining whether Glauber dynamics is stable against noise or errors is therefore an important question and, as far as we are aware, still open (with partial results obtained under the assumption of attractiveness \cite{MR814713}).

In this section, we present a natural embedding of Glauber dynamics into the Linbdlabian setting, showing how this embedded dynamics inherits properties from the classical Markov chain\footnote{
A similar construction was proposed in \cite{2010NJPh...12b5021A}.}.
We will then apply the results of section~\ref{sec:main-result} to prove, in the appropriate regime, stability of Glauber dynamics.

We will consider a lattice spin system over $\Gamma = \Z^D$ or $\Gamma = (\Z/L\Z)^D$, with (classical) configuration space of a single spin a finite set $S$. For simplicity, we will consider the case $S=\{+1,-1\}$. For each $\Lambda \subset \Gamma$, we will denote by $\Omega_\Lambda$ the space of configurations over $\Lambda$, namely $S^\Lambda$. $\Lambda^c$ will denote the complementary of $\Lambda$ in $\Gamma$, namely $\Gamma \setminus \Lambda$.
\begin{definition}
A finite range, translationally-invariant potential $\{J_A\}_{A\subset \Gamma}$ is a family of real functions indexed by the non empty finite subsets of $\Gamma$ satisfying the following properties:
\begin{enumerate}
\item $J_A : \Omega_A \to \R$.
\item For all $A\subset \Gamma$ and all $x \in \Gamma$:
  \[ J_A(\sigma) = J_{A+x}(\eta) \quad \text{if}\quad \sigma(y+x) = \eta(y)\quad \forall y \in A.\]
\item There exists a positive $r >0$ such that $J_A=0$ if $\diam A > r$, called the range of interaction.
\end{enumerate}

Given a finite-range, translationally-invariant potential, we can define a Hamiltonian for each finite lattice $\Lambda \subset \Gamma$ and each \emph{boundary condition} $\tau \in \Omega_{\Lambda^c}$ by
\[ H_\Lambda^\tau (\sigma) = - \sum_{A \cap \Lambda \neq 0} J_A(\sigma \times \tau ) \quad \forall \sigma \in \Omega_\Lambda ,\]
where $\sigma \times \tau$ is the configuration that agrees with $\sigma$ over $\Lambda$ and with $\tau$ over $\Lambda^c$.
For each such Hamiltonian, we define the Gibbs state state as
\[ \mu_\Lambda^\tau (\sigma) = {(Z_\Lambda^\tau)}^{-1} \exp (- H_\Lambda^\tau(\sigma)) ,\] where $Z_\Lambda^\tau$ is a normalizing constant.\footnote{Following \cite{Martinelli97}, in our notation we have incorporated the usual inverse temperature parameter $\beta$ directly into the potential $J$.} The convex hull of the set of Gibbs states over $\Lambda$ will be denoted by $\mathcal G(\Lambda)$:
\[ \mathcal G(\Lambda) = \operatorname{conv} \{ \mu_\Lambda^\tau \,|\, \tau \in \Omega_{\Lambda^c} \} .\]
\end{definition}

\begin{definition}
  The Glauber dynamics for a potential $J$ is the Markov process on $\Omega_\Lambda$ with the following generator:
  \[ (Q_\Lambda f) (\sigma) = \sum_{x \in \Lambda} c_J(x, \sigma) \nabla_x f (\sigma) ,\]
  where $\nabla_x f(\sigma)$ if defined as $f(\sigma^x) - f(\sigma)$, and $\sigma^x$ is the configuration obtained by flipping the spin at position $x$:
  \[ \sigma^x(y) = \begin{cases}
    \sigma(y) &\text{if } x\neq y\\
    -\sigma(x) &\text{if } x = y.
  \end{cases}\]

  The numbers $c_J(x,\sigma)$ are called transition rates and must satisfy the following assumptions:
  \begin{enumerate}
  \item Positivity and boundedness: There exist positive constants $c_m$ and $c_M$ such that:
    \[  0 < c_m \le  c_J(x,\sigma) \le c_M < \infty \quad \forall x,\sigma.\]
  \item Finite range: $c_J(x, \cdot)$ depends only on spin values in $b_r(x)$.
  \item Translational invariance: for all $k \in \Gamma$,\[ c_J(x,\sigma^\prime) = c_J(x+k,\sigma)\quad\text{if}\quad \sigma^\prime(y) = \sigma(y+k)\quad\forall y .\]
  \item Detailed balance: for all $x \in \Gamma$ and all $\sigma$
    \[ \exp\left(-\sum_{A \ni x } J_A(\sigma)\right) c_J(x,\sigma) = c_J(x,\sigma^x)\exp\left(-\sum_{A\ni x}J_A(\sigma^x)\right) .\]
  \end{enumerate}
  These assumptions are sufficient to ensure that $Q_\Lambda$ generates a Markov process which has the Gibbs states over $\Lambda$ as stationary points.
\end{definition}

\begin{definition}
A quantum embedding of the classical Glauber dynamics for a potential $J$ is generated by the following Lindblad operators
\begin{align}
L_{x,\eta} = \sqrt{c_J(x,\eta)} \ketbra{}{\eta^x}{\eta}\otimes \identity, &\quad \forall x \in \Lambda, \forall \eta \in \Omega_{b_x(r)} ; \\
\mathcal L_{x,\eta}(\rho) = L_{x,\eta} \rho \du{L_{x,\eta}} &- \frac 12 \left\{\rho, c_J(x,\eta) \diagstate{\eta} \right\} ; \nonumber \\
\mathcal L_\Lambda (\rho) = \sum_{x\in \Lambda} \sum_{\eta} L_{x,\eta} \rho \du{L_{x,\eta}} - \frac12 \{ \rho, K \}, &\quad K = \sum_{\sigma} \left(\sum_x  c_J(x,\sigma)\right) \diagstate{\sigma} ;
\end{align}
plus a dephasing channel acting independently and uniformly on all sites $x\in \Lambda$:
\begin{equation}
\label{eq:dephasing}
D_{x,0} = \sqrt{\gamma} \ketbra{}{0}{0}, \quad D_{x,1} = \sqrt{\gamma} \ketbra{}{1}{1}, \quad \mathcal D(\rho) =
\sum_{x\in \Lambda} \sum_{i=0,1}D_{x,i} \rho \du D_{x,i} - \abs{\Lambda}\gamma \rho.
\end{equation}
\end{definition}

$\mathcal L_\Lambda$ satisfies translational invariance because the transition rates $c_J$ do,
and it easy to see that this family of Lindbladians is uniform.

\begin{remark}
Take $\ketbra{}{\alpha}{\beta}$ an element of the computational basis, and let $d(\alpha,\beta)$ be the Hamming distance between $\alpha$ and $\beta$. Then it holds that
\[ \mathcal D(\ketbra{}{\alpha}{\beta}) = -\gamma d(\alpha, \beta) \ketbra{}{\alpha}{\beta} .\]
In other words, $\mathcal D$ is a Schur multiplier in the computational basis, represented by $\left(-\gamma d(\alpha,\beta) \right)_{\alpha,\beta}$.

On the other hand, we have that for all $x$:
\begin{equation}
\label{eq:quantum-glauber-action}
 \sum_{\eta \in \Omega_{b_x(r)}} \mathcal L_{x,\eta} (\ketbra{}{\alpha}{\beta}) = \begin{cases}
c_J(x,\alpha) \left( \ketbra{}{\alpha^x}{\beta^x} - \ketbra{}{\alpha}{\beta} \right) & \text{if } \alpha|_{b_x(r)} = \beta|_{b_x(r)} ,\\\\
- \frac 12 \left( c_J(x,\alpha) + c_J(x, \beta) \right) \ketbra{}{\alpha}{\beta} & \text{otherwise}.
\end{cases}
\end{equation}
Since $d(\alpha^x,\beta^x) = d(\alpha,\beta)$, $[ \mathcal D, \sum_{\eta} \mathcal L_{x,\eta} ] = 0$ for all $x \in \Lambda$, and in particular $\mathcal D$ and $\mathcal L_\Lambda$ commute.
\end{remark}

This quantum dissipative system inherits various properties from its classical counterpart.

\begin{definition}
Let $\mu$ be a full-rank positive state. Denote by
\[ \Gamma_\mu(\rho) = \mu^{\frac12} \rho \mu^{\frac12} .\]
We say that $\mathcal L$ is in detailed balance \cite{temme:122201,MR0468989,MR737310,MR1647879} with respect to $\mu$ if $ \Gamma_\mu \circ \mathcal L = \du{\mathcal L} \circ \Gamma_\mu$.
\end{definition}

\begin{proposition}
  Let $\mu_\Lambda^\tau$ be a Gibbs state over $\Lambda$. Then $\mathcal L_\Lambda$ and $\mathcal D$ are in detailed balance with respect to $\mu_\Lambda^\tau$.
\end{proposition}
\begin{proof}
  Note that $\Gamma_{\mu_\Lambda^\tau}$ is a Schur multiplier in the computational basis:
  \[ \Gamma_{\mu_\Lambda^\tau}(\ketbra{}{\eta_1}{\eta_2} )= \mu_\Lambda^\tau(\eta_1)^{\frac12} \mu_\Lambda^\tau(\eta_2)^{\frac12} \ketbra{}{\eta_1}{\eta_2} .\]
  From the detailed balance condition for the transition rates $c_J(x,\sigma)$, it follows that for all $x\in \Lambda$, denoting $\mathcal L_x = \sum_{\eta \in \Omega_{b_x(r)} } \mathcal L_{x,\eta}$,
  \begin{align*}
    \Gamma_{\mu_\Lambda^\tau} &\circ \mathcal L_x \circ \Gamma_{\mu_\Lambda^\tau}^{-1} (\ketbra{}{\eta_1}{\eta_2} ) \\
    &= \delta^x_{\eta_1,\eta_2} \left( c_J(x,\eta_1)\frac{\mu_\Lambda^\tau(\eta_1^x)}{\mu_\Lambda^\tau(\eta_1)}  \right)\ketbra{}{\eta_1^x}{\eta_2^x} -  \frac{ c_J(x,\eta_1) + c_J(x,\eta_2)}{2} \ketbra{}{\eta_1}{\eta_2} \\
    &= \delta^x_{\eta_1,\eta_2} c_J(x,\eta_1^x)\ketbra{}{\eta_1^x}{\eta_2^x} - \frac{ c_J(x,\eta_1) + c_J(x,\eta_2)}{2} \ketbra{}{\eta_1}{\eta_2} \\
    &= \du{\mathcal L_x} (\ketbra{}{\eta_1}{\eta_2} )),
  \end{align*}
  where
  \[ \delta^x_{\eta_1,\eta_2} = \begin{cases} 1 &\text{if } \eta_1|_{b_x(r)} = \eta_2|_{b_x(r)} \\ 0 & \text{otherwise}.
  \end{cases}\]

  To prove detailed balance for $\mathcal D$, note that Schur multipliers commute, thus $[ \mathcal D, \Gamma_\mu ] =0$. This, together with the fact that $\du {\mathcal D} = \mathcal D$, implies that $\mathcal D$ is in detailed balance w.r.t. $\mu_\Lambda^\tau$.
\qed
\end{proof}

The above proposition implies that Gibbs states are stationary states for the quantum Glauber dynamics. Let us prove that there are no other fixed points apart from the classical ones (i.e. states that are diagonal in the computational basis). Clearly, $\mathcal D$ has all classical states as stationary points. We just have to check $\mathcal L_\Lambda$.
\begin{proposition}
  The set of fixed points of $\mathcal L_\Lambda$ is equal to $\mathcal G(\Lambda)$, the set of Gibbs states over $\Lambda$.
\end{proposition}
\begin{proof}
  Let $\rho$ be a fixed point of $\mathcal L_\Lambda$. We want to prove that $\rho$ is diagonal, i.e.\ that it is of the form
  \[ \rho = \sum_{\sigma} p_\sigma \diagstate{\sigma} .\]

  Consider a non-diagonal element $\ketbra{}{\alpha}{\beta}$, and suppose $\alpha(x)\neq\beta(x)$ for some $x \in \Lambda$. Then, from equation~\eqref{eq:quantum-glauber-action}, we have that for all $y\in b_x(r)$,
  \[\mathcal L_y  (\ketbra{}{\alpha}{\beta}) = -\frac12 ( c_J(y,\alpha) + c_J(y,\beta))  \ketbra{}{\alpha}{\beta} .\]
  For $y\nin b_x(r)$, $\mathcal L_y$ is not supported on $x$, and thus cannot change the configuration there. This implies that the evolution cannot change the configurations over the set $\Delta(r)$, where $ \Delta = \{ x \in \Lambda \,|\, \alpha(x) \neq \beta(x) \}$. In turn, this implies that $\mathcal L_{\Delta}$ commutes with $\mathcal L - \mathcal L_{\Delta}$ (since it acts as a Schur multiplier whose entries depend only on the sites in $\Delta(r)$). Finally, this means that
  \begin{align*}
    \norm{ e^{t \mathcal L_\Lambda}  (\ketbra{}{\alpha}{\beta} )}_1
    &\le \norm{e^{t \mathcal L_{\Delta}} (\ketbra{}{\alpha}{\beta} )}_1
    = \exp\left(-t\frac12 \left(\sum_{x\in\Delta} c_J(x,\alpha) + c_J(x,\beta) \right) \right)\\
    &\le \exp \left( - t \frac12 c_m d(\alpha,\beta) \right) \to 0 .
  \end{align*}
  Since the off-diagonal elements are killed, $\rho$ must be of the form $\sum_\sigma p_\sigma \diagstate{\sigma}$. Writing the equation $\mathcal L_\Lambda(\rho) = 0$ we obtain
  \[ \sum_\sigma \sum_x c_J(x,\sigma) p_\sigma \diagstate{\sigma^x} - \sum_\sigma \sum_x c_J(x,\sigma) p_\sigma \diagstate{\sigma} = 0 ,\]
  which implies
  \[ \sum_x c_J(x, \sigma^x) p_{\sigma^x} = \sum_x p_\sigma c_J(x,\sigma) .\]
  The last equation is simply a rewriting of the fact that $(p_\sigma)$ is a stationary distribution for $Q_\Lambda$, that is, it is exactly a Gibbs state on $\Lambda$.
\qed
\end{proof}

Since $\mathcal L_\Lambda$ and $\mathcal L_\Lambda + \mathcal D$ have the same stationary distributions, even locally, all properties that depend just on the structure of the fixed-point sets will be shared by both: this is the case, for example, of frustration freeness (which we will prove next) and LTQO (which will be proved later).

\begin{proposition}
\label{prop:glauber-ff}
$\mathcal L_\Lambda$ (and consequently $\mathcal L_\Lambda +\mathcal D$) is frustration free.
\end{proposition}
\begin{proof}
By the previous proposition, we have that $\mathcal X_{\mathcal L_\Lambda} = \mathcal G(\Lambda)$. We know \cite{Liggett85} that for Gibbs states it holds that
\[ \Delta \subset \Lambda \Rightarrow \mathcal G(\Lambda) \subset \mathcal G(\Delta) ,\]
but this is exactly the frustration-freeness condition for $\mathcal L_\Lambda$.
\qed
\end{proof}

\subsection{Stability of Glauber dynamics}
We want to show that the contraction of the semigroup generated by $\mathcal L_\Lambda + \mathcal D$ can be controlled by the contraction of the classical Glauber dynamics. To fix notation, denote by $\mathcal C : \mathcal A_\Lambda \to \mathcal A_\Lambda$ the projector on the diagonal subspace with respect to the computational basis. $\mathcal C$ is a completely positive, trace preserving map, and it also satisfies $\mathcal C = \lim_{t \to \infty} \exp( t \mathcal D )$. Since $\mathcal L_\Lambda$ commutes with $\mathcal D$, it also commutes with $\mathcal C$. Then we can prove the following:
\begin{lemma}
  \label{lemma:dephasing-noise}
  If $T_t = \exp\left(t (\mathcal L_\Lambda + \mathcal D) \right)$, then
  \begin{equation}
    \eta(T_t) \le \eta(T_t \circ \mathcal C ) + \eta( \exp(t \mathcal D) ) .
  \end{equation}
\end{lemma}

\begin{proof}
  Fix an initial state $\rho$. Then we can write
  \begin{align*}
    \norm{ T_t(\rho) - T_\infty(\rho) }_1
    &\le \norm{ T_t\circ \mathcal C (\rho) - T_\infty(\rho) }_1 + \norm{T_t \circ (1-\mathcal C)(\rho) }_1 \\
    &\le \norm{ T_t\circ \mathcal C (\rho) - T_\infty \circ \mathcal C(\rho) }_1 + \norm{\exp(t\mathcal D)\circ (1-\mathcal C)(\rho)  }_1 \\
    &\le \eta(T_t \circ \mathcal C ) + \eta( \exp(t \mathcal D) ),
  \end{align*}
  where we have used the fact that $\mathcal L_\Lambda$ and $\mathcal D$ commute, and that the fixed points of $\mathcal L_\Lambda$ are invariant under $\mathcal C$.
\qed
\end{proof}

We know, because of theorem~\ref{thm:commuting-liouvillians}, that
\begin{equation}
  \eta(\exp(t\mathcal D)) \le \abs{\Lambda} e^{-\frac{\gamma}{2} t} ,
\end{equation}
and this implies the following  result.
\begin{corollary}
  If the classical Glauber dynamics satisfies rapid mixing, then also the quantum embedded Glauber dynamics generated by $\mathcal L_\Lambda + \mathcal D$ does.
\end{corollary}

\begin{remark}
  Convergence rates of classical Glauber dynamics are a well studied subject. It is known that, in some regimes, classical Glauber dynamics satisfies a Log Sobolev inequality with system-size independent Log Sobolev constant (for a review on the subject see \cite{Martinelli97}). In such situations the classical chain has a logarithmic mixing time, and thus satisfies rapid mixing.
\end{remark}

For this class of classical dynamical systems it is possible to apply our main result~\ref{thm:stability}. In particular, we can arbitrary perturb the transition rates $c_J(x,\sigma)$ by some $e(x,\sigma)$, not necessary preserving detailed balance. If we denote by $\mathcal E$ the maximum of $\abs{e(x,\sigma)}$, the difference between the perturbed and the original evolution of local observables can be bounded by $\mathcal E$ times a factor depending on the size of the support of the observables taken into account.

\begin{theorem}
\label{thm:stability-glauber}
Let $Q_\Lambda$ the generator of a classical Glauber dynamics, having a unique fixed point and satisfying a Log Sobolev inequality with constant independent of system size. Let $E$ be the generator of another classical Markov process of the form
\[ (E f)(\sigma) = \sum_{x \in \Lambda} e(x,\sigma) \nabla_x f(\sigma) .\]
Suppose that $\mathcal E = \sup_{x,\sigma} \abs{e(x,\sigma)} < \infty $ and that $e(x,\cdot)$ has support bounded uniformly in $x$. Denote by $T_t$ the evolution generated by $Q_\Lambda$ and by $S_t$ the evolution generated by $Q_\Lambda + E$. Then, for each function $f$ supported on $A \subset \Lambda$, it holds that
\[ \norm{T_t(f) - S_t(f)}_\infty \le c(\abs{A}) \norm{f}_\infty \mathcal E ,\]
for some $c(\cdot)$ independent of system size and polynomially growing.
\end{theorem}

\begin{remark}
It is known \cite{martinelli-2d,Lubetzky-Sly} that the Ising model on $\Z^2$ or $(\Z/n\Z)^2$ has a system size independent Log Sobolev constant for high temperatures (when the inverse temperature $\beta$ is lower than the critical value $\beta_c$), or at any temperature in presence of an external magnetic field. In this regime the Glauber dynamics sampling the Ising model is stable (in the sense of theorem~\ref{thm:stability}).
\end{remark}

\subsection{Weak mixing and LTQO}
As a nice observation, though not necessary to prove theorem~\ref{thm:stability-glauber},
we want show that \emph{weak mixing}, a condition on Gibbs states defined in
\cite{Martinelli97}, is equivalent to the LTQO condition given in section~\ref{sec:proof}. The weak mixing
conditions for two-dimensional systems has been shown \cite{martinelli-2d} to imply $L_2$ convergence of the
corresponding Glauber dynamics.

\begin{definition}
  We say that the Gibbs measures in $\mathcal G(\Lambda)$ satisfy the \emph{weak mixing} condition in $V\subset \Lambda$ if there exist constants $C$ and $m$ such that, for every subset $\Delta \subset V$, the following holds:
  \begin{equation}
    \sup_{\tau, \tau^\prime \in \Omega_{V^c}} \norm{ \mu^\tau_{V,\Delta} - \mu^{\tau^\prime}_{V,\Delta} }_1 \le C \sum_{\substack{ x \in \Delta,\\ y \in \partial_r^+ V}} e^{-m \dist(x,y) },
  \end{equation}
  where $\partial_r^+ V = \{ x \in V^c \,|\, \dist(x,V) \le r \}$ and $\mu^\tau_{V,\Delta} = \trace_{V\setminus \Delta} \mu^\tau_V$.
\end{definition}

\begin{proposition}
  If $\mathcal G(\Lambda)$ satisfies the weak mixing condition for each $V \subset \Lambda$, then $\mathcal L_\Lambda$ (and consequently $\mathcal L_\Lambda +\mathcal D$) satisfies LTQO.
\end{proposition}
\begin{proof}
  Take $A\subset \Lambda$, $\ell \ge 0$, and let $V$ be $A(\ell)$. The weak mixing condition for $V$ implies that there exist constants $C$ and $m$ such that
  \[ \sup_{\tau, \tau^\prime \in \Omega_{V^c}} \norm{ \mu^\tau_{V,A} - \mu^{\tau^\prime}_{V,A} }_1 \le C \sum_{\substack{ x \in A,\\ y \in \partial_r^+ V}} e^{-m \dist(x,y) } \le C e^{-m \ell} \abs{A} \abs{\partial_r^+ A(\ell)} .\]
  This is the LTQO condition with $\Delta_0 (\ell) =  C e^{-m \ell} \abs{A} \abs{\partial_r^+ A(\ell)}$. The bound, proven for states of the form $\mu_V^\tau$, can be extended by convexity to all $\mathcal G(V)$. Let $\eta_0, \eta_1 \in \mathcal G(V)$. By definition, $\eta_0$ and $\eta_1$ are convex combination of states of the form $\mu_V^\tau$, thus we can write
  \[ \eta_0 = \sum_i p_i \mu_V^{\tau_i}, \quad \eta_1 = \sum_j q_j \mu_V^{\sigma_j}, \quad \sum_i p_i = \sum_j q_j = 1; \quad p_i, q_j \ge 0.\]
  Then we have
  \begin{align*}
    \norm{\eta_{0,A} - \eta_{1,A}}_1 &= \norm{ \sum_i p_i \mu_{V,A}^{\tau_i} - \sum_j q_j \mu_{V,A}^{\sigma_j} }_1 \\
    &= \norm{ \sum_i p_i (\sum_j q_j \mu_{V,A}^{\tau_i} ) - \sum_j q_j (\sum_i p_i \mu_{V,A}^{\sigma_j} ) }_1 \\
    &\le \sum_{i,j} p_i q_j \norm{ \mu_{V,A}^{\tau_i} - \mu_{V,A}^{\sigma_j}}_1
    \le \sup_{\tau, \sigma} \norm{ \mu_{V,A}^\tau - \mu_{V,A}^\sigma}_1 .
  \end{align*}
\qed
\end{proof}


\section{Conclusions and open questions}
\label{sec:conclusions}

In the context of local perturbations of local Hamiltonians,
changes in the ground state can be detected by the lack of
smoothness of the expectation value of local observables.
Via the quasi-adiabatic technique \cite{quasi-adiabatic-2},
the regularity of such expectation values can be
related to the study of the effect that the perturbation has on the
spectral gap of the Hamiltonian.
In \cite{Spiros11}, the stability of the spectral gap was shown under the assumptions
of frustration-freeness and local indistinguishability between ground states of local patches
of the original Hamiltonian.

In this paper we have studied a class of open quantum systems
described by local Lindbladian evolutions with unique fixed points,
focusing on the problem of the smoothness of evolution of local observables
in the presence of local perturbations.
Given any initial configuration, the system will converge toward
the fixed point with a certain rate. The slowest rate over all possible
initial configurations defines a mixing property of the Lindbladian,
and we consider how this scales with the system size.
In the case of power-law decay of interactions, we show
that a logarithmic scaling is sufficient for the stability of the evolution of local observables,
while for exponentially decaying and finite range interactions a scaling
at least as fast as a certain polynomial, determined by equation~\eqref{eq:cc-relaxed-mix}, is also sufficient.
Moreover, the same assumptions imply certain properties of the fixed point, such as local topological quantum order.
It should be emphasized that Log Sobolev inequalities provide strong enough convergence-time estimates to
satisfy our assumptions, but that our results also apply more generally.

The most important open question involves state engineering of degenerate topologically ordered states, such as topologically protected quantum codes. For such states, all known preparation maps have a convergence time that is slower than required for our result to apply \cite{2013arXiv1310.1037K}.
It is an interesting question whether it is possible to exploit the very weak requirements in terms of locality
of the boundary condition in our definition of uniform families (see definition~\ref{defn:boundary-condition})
to construct faster mixing maps for which one could prove stability, since logical observables partially
supported on such boundaries are not necessarily localizable in the sense of \cite{2013arXiv1310.1037K}.


\begin{paragraph}{Acknowledgments}
  T.S.C. is supported by a Royal Society University Research fellowship, and was previously supported by a Juan de la Cierva fellowship.
  T.S.C., A.L., and D.P.-G. are supported by Spanish grants MTM2011-26912 and QUITEMAD, and European CHIST-ERA project CQC (funded partially by MINECO grant PRI-PIMCHI-2011-1071).
  A.L. is supported by Spanish Ministerio de Economía y Competividad FPI fellowship BES-2012-052404.
 SM acknowledges funding provided by the Institute for
Quantum Information and Matter, an NSF Physics Frontiers Center with support of the Gordon and Betty Moore
Foundation through Grant \#GBMF1250 and by the AFOSR Grant \#FA8750-12-2-0308.
  The authors would like to thank the hospitality of the Centro de Ciencias Pedro Pascual in Benasque, where part of this work was carried out.
\end{paragraph}

\begin{appendices}
\section{The non-stable example}
\label{sec:non-stable-example}
The following example will satisfy all the conditions of theorem~\ref{thm:stability}, except forming an uniform family, and will be shown to be unstable. Interestingly, the system \emph{is} rapid mixing, showing that without the correct structure with respect to system size scaling, rapid mixing alone is not sufficient to imply stability of local observables. This example is the generalization to dissipative systems of the globally gapped but not locally gapped example in \cite{Spiros11}. We will show that the characteristics of the dynamics are essentially determined by a classical Markov chain embedded into the Lindbladian. For a general review on convergence of Markov chains, see \cite{Levin-Peres}.

\begin{example}
  Consider a chain of $2N$ classical spins, with values in $\{0,1\}$.
  Let us define a generator $Q^{2N}$ of a classical Markov chain over the configuration~space~$\{0,1\}^{2N}$.
  We will define $Q^{2N}$ in a translationally-invariant way as follows:
  \begin{alignat*}{2}
   Q_c &=       \bordermatrix{
     ~ & \ket{10}         & \ket{00}   & \ket{11} & \ket{01} \cr
     \ket{10} & -\frac{2}{3N}  & 0             & 0            & \frac{2}{3N}  \cr
     \ket{00} & 0                    & -1           & 0             & 1 \cr
     \ket{11} & 0                    & 0             & -1          & 1 \cr
     \ket{01} & 0                    & 0             & 0            & 0 \cr
   }
   , \quad &
   Q_r &=       \bordermatrix{
     ~ & \ket{10} & \ket{00}         & \ket{11} & \ket{01} \cr
     \ket{10} & -1          & 0                    & 1            & 0  \cr
     \ket{00} & 0            & -1                  & 0            & 1 \cr
     \ket{11} & 0            & 0                    & 0            & 0 \cr
     \ket{01} & 0            & 0                    & 0            & 0 \cr
   }
   , \\
   Q_l &=       \bordermatrix{
     ~ & \ket{10} & \ket{00}          & \ket{11} & \ket{01} \cr
     \ket{10} & -1          & 1                    & 0            & 0  \cr
     \ket{00} & 0            & 0                    & 0            & 0 \cr
     \ket{11} & 0            & 0                    & -1          & 1 \cr
     \ket{01} & 0            & 0                    & 0            & 0 \cr
   }
   , \quad &
   \delta_0 &= \diagstate{0}, \quad \delta_1 = \diagstate{1}.
 \end{alignat*}
 We then define for each $i=1 \dots N$, a generator matrix $Q_i$ acting on spins $(2i-2, \dots, 2i+1)$ by
 \[ Q_i = \identity \otimes Q_c \otimes \identity + \identity \otimes Q_r \otimes \delta_0 + \delta_1 \otimes Q_l \otimes \identity ;\]
 and $Q^{2N} = \sum_{i=1}^{N} Q_i$.

 The matrix $Q_i$ can only change spins $(2i-1,2i)$: its transition graph restricted to such
 spins is presented in figure~\ref{fig:transition-matrix-eg1}.

  \begin{figure}[H]
   \centering
   \begin{minipage}[c]{0.5\textwidth}
     \begin{tikzpicture}
       \tikzstyle{pair}=[shape=circle,draw=orange!50,fill=orange!20]
       \node[pair] (00) at (0,2) [label=above:$00$] {};
       \node[pair] (01) at (0,0) [label=below:$01$] {};
       \node[pair] (10) at (2,2) [label=above:$10$] {};
       \node[pair] (11) at (2,0) [label=below:$11$] {};
       \draw[->] (00) to node[auto,near start] {}  (01);
       \draw[->] (11) to node[auto,swap,near start] {}  (01);
       \draw[->] (10) to node[auto,pos=0.35] {$\frac{2}{3N}$}  (01);
       \draw[->,bend right,draw=blue] (00) to node[auto,swap] {$\delta_0(2i+1)$}  (01);
       \draw[->,bend left,draw=blue] (10) to node[auto] {$\delta_0(2i+1)$}  (11);
       \draw[->,bend right,draw=red] (10) to node[auto,swap] {$\delta_1(2i-2)$}  (00);
       \draw[->,bend left,draw=red] (11) to node[auto] {$\delta_1(2i-2)$}  (01);
     \end{tikzpicture}
   \end{minipage}
   \begin{minipage}[c]{0.4\textwidth}
     \begin{equation*}
       \bordermatrix{
         ~ & \ket{10} & \ket{00}          & \ket{11} & \ket{01} \cr
         \ket{10} & *            & {\color{red} 1} & {\color{blue} 1} & \frac{2}{3N} \cr
         \ket{00} & 0            & *                     & 0                     & 1 + {\color{blue} 1} \cr
         \ket{11} & 0            & 0                    & *            & 1 + {\color{red} 1} \cr
         \ket{01} & 0            & 0                    & 0            & 0 \\
       }
     \end{equation*}
   \end{minipage}
   \caption{The transition matrix for $Q_i$ on the spins $(2i-1,2i)$. The blue and the red transitions are present depending on the nearby sites: the blue ones if there is a $0$ on the right, the red ones if there is a $1$ on the left. Asterisks in the diagonal are such that the sum of each row is zero.}
   \label{fig:transition-matrix-eg1}
 \end{figure}
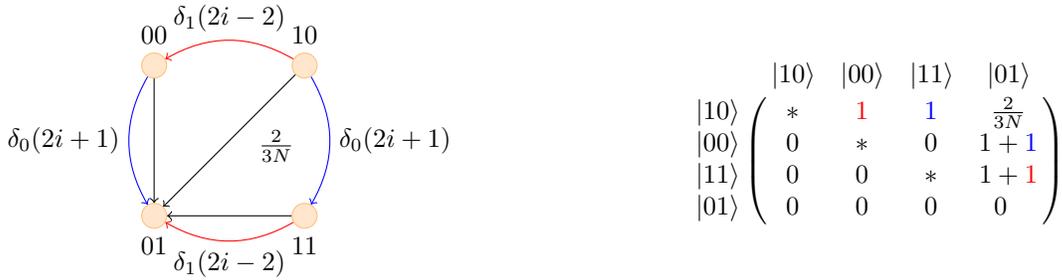

 By construction, $Q^{2N}$ is upper triangular. Thus the elements on the diagonal are the eigenvalues.
 The unique steady state is then $\ket{0101\dots01}$, and the smallest non-zero eigenvalue,
 corresponding to the state $\ket{1010\dots10}$, is $\frac{2}{3}$.
 Furthermore, it is easy to see that the diameter of the graph of the transitions of $Q^{2N}$ is $N$,
 and in turn this implies that the mixing time for $ Q^{2N}$ is of order $O(\log N)$\footnote{
   This can be seen from the upper triangular form of $Q^{2N}$, noticing that the polynomials
   appearing in $e^{t Q^{2N}}$ have degree of at most the diameter of the transition graph.}.

 Let us now embed this classical Markov chain into a Lindbad operator, in a similar fashion as we have done in section~\ref{sec:glauber-dynamics}
 with Glauber dynamics. We will consider then a chain of $2N$ qubits, and define the following Lindblad operators:
  if $k$ is odd, then
 \begin{align*}
   L_{k,1} &= \sigma_x^{k+1}  \ketbra k 0 0 \otimes \ketbra {k+1} 0 0  ,\\
   L_{k,2} &= \sigma_x^{k\phantom{+1}}  \ketbra k 1 1 \otimes \ketbra {k+1} 1 1  ,\\
   L_{k,3} &= \sqrt{\frac{2}{3N}} \sigma_x^k \otimes \sigma_x^{k+1}   \ketbra k 1 1 \otimes \ketbra {k+1} 0 0 ;
 \end{align*}
 if $k$ is even, then
 \begin{align*}
   L_{k,1} &= \sigma_x^{k\phantom{+1}}  \ketbra k 0 0 \otimes \ketbra {k+1} 0 0  ,\\
   L_{k,2} &= \sigma_x^{k+1}  \ketbra k 1 1 \otimes \ketbra {k+1} 1 1  ,\\
   L_{k,3} &= 0.
 \end{align*}
 The Lindbladian is then defined translationally-invariantly as
 \[ \mathcal L^{2N} = \sum_{k=1}^{2N} \sum_{i=1}^3 \mathcal L_{k,i} + \mathcal D_k ;\]
 where $\mathcal D_k$ is a dephasing channel acting on site $k$, as in equation~\eqref{eq:dephasing}.
 Since $L_{k,3}$ depends on $N$, the family we have defined is not a uniform family.

 It is easy to see that the action of $\mathcal L^{2N}$ on diagonal states of the form $\diagstate{\alpha}$,
 with $\alpha \in \{0,1\}^{2N}$, is equal to that of $Q^{2N}$ acting on $\alpha$: this is indeed an embedding
 of $Q^{2N}$.

 Then, by a similar argument as in section~\ref{sec:glauber-dynamics}, we can prove that the fixed points of $\mathcal L^{2N}$ are
 exactly the same as those of $Q^{2N}$ (namely, the unique state\linebreak[4] $\diagstate{0101 \dots 01}$), and that the mixing time of $\mathcal L^{2N}$
 is bounded by the sum of the mixing times of $Q^{2N}$ and of $\mathcal D$. Since both of them are mixing in time $O(\log N)$,
 we see that $\mathcal L^{2N}$ satisfies rapid mixing.

 But the system is unstable: if we perturb $\mathcal L^{2N}$ by removing the terms generated by $L_{k,3}$
 (which is a perturbation of order $O(\frac{1}{N})$), the diagonal state $\diagstate{1010 \dots 10}$ becomes a
 stationary state, and it is clearly locally ortogonal from the original one $\diagstate{0101 \dots 01}$.
\end{example}


\end{appendices}

\printbibliography
\end{document}